\documentclass[letterpaper,11pt]{article}

\usepackage{amsmath,amsthm,amssymb,amsfonts}

\usepackage[mono=false]{libertine}
\usepackage{libertinust1math}
\usepackage[T1]{fontenc}

\usepackage[margin=1in]{geometry}
\usepackage{subfiles}
\usepackage{hyperref}
\hypersetup{colorlinks=true, citecolor=blue, linkcolor = black, urlcolor=blue}
\usepackage{enumerate}
\usepackage[shortlabels]{enumitem} 
\usepackage{graphicx}
\graphicspath{{figures/}}
\usepackage{subcaption}
\usepackage{titling}
\usepackage{multirow}
\usepackage{array}
\usepackage{booktabs}
\usepackage{wasysym} 
\usepackage{algorithm2e}
\SetKwInOut{Problem}{Problem}
\SetKwInOut{Instance}{Instance}
\usepackage{pbox} 

\newtheorem{theorem}{Theorem}[section]
\newtheorem*{theorem*}{Theorem}
\newtheorem{corollary}[theorem]{Corollary}
\newtheorem{lemma}[theorem]{Lemma}
\newtheorem*{lemma*}{Lemma}
\newtheorem{proposition}[theorem]{Proposition}

\theoremstyle{definition}
\newtheorem{definition}{Definition}[section]
\newtheorem*{definition*}{Definition}
\theoremstyle{remark}
\newtheorem{remark}{Remark}[section]
\newtheorem*{notation}{Notation}

\numberwithin{equation}{section}

\newcommand{\figref}[1]{Figure~\ref{#1}}
\newcommand{\secref}[1]{Section~\ref{#1}}
\newcommand{\thmref}[1]{Theorem~\ref{#1}}
\newcommand{\lemref}[1]{Lemma~\ref{#1}}
\newcommand{\propref}[1]{Proposition~\ref{#1}}
\newcommand{\defnref}[1]{Definition~\ref{#1}}
\newcommand{\corref}[1]{Corollary~\ref{#1}}

\newcommand{\appref}[1]{Appendix~\ref{#1}}

\date{}

\begin{document}

\title{Counting perfect matchings and the eight-vertex model}
\author{
{Jin-Yi Cai}\thanks{Department of Computer Sciences, University of Wisconsin-Madison. Supported by NSF CCF-1714275. \texttt{jyc@cs.wisc.edu}}
\and {Tianyu Liu}\thanks{Department of Computer Sciences, University of Wisconsin-Madison. Supported by NSF CCF-1714275. \texttt{tl@cs.wisc.edu}}
}

\maketitle

\hypersetup{colorlinks=true, citecolor=blue, linkcolor=red, urlcolor=blue}
\setcounter{footnote}{2}

\begin{abstract}
We study the approximation complexity of the partition function of the eight-vertex model on general 4-regular graphs. For the first time, we relate the approximability of the eight-vertex model to the complexity of approximately counting perfect matchings, a central open problem in this field.
Our results extend those in \cite{DBLP:journals/corr/abs-1811-03126}.

In a region of the parameter space where no previous approximation complexity was known,
we show that approximating the partition function is at least as hard as approximately counting perfect matchings via approximation-preserving reductions.
In another region of the parameter space which is larger than the previously known FPRASable region, we show that computing the partition function can be reduced to (with or without approximation) counting perfect matchings.
Moreover, we give a complete characterization of nonnegatively weighted (not necessarily planar) 4-ary matchgates, which has been open for several years.
The key ingredient of our proof is a \emph{geometric} lemma.

We also identify a region of the parameter space where approximating the partition function on \emph{planar} 4-regular graphs is feasible but on \emph{general} 4-regular graphs is equivalent to approximately counting perfect matchings. To our best knowledge, these are the first problems of this kind.
\end{abstract}


\section{Introduction}\label{sec:intro}

The eight-vertex model is defined over 4-regular graphs, the states of which are the set of \emph{even orientations}, i.e. those with an even number of arrows into (and out of) each vertex.
There are eight permitted types of local configurations around a vertex---hence the name eight-vertex model (see \figref{fig:orientations}).

\renewcommand{\thesubfigure}{-\arabic{subfigure}}
\begin{figure}[h!]
\centering
\begin{subfigure}[b]{0.12\linewidth}
\centering\includegraphics[width=\linewidth]{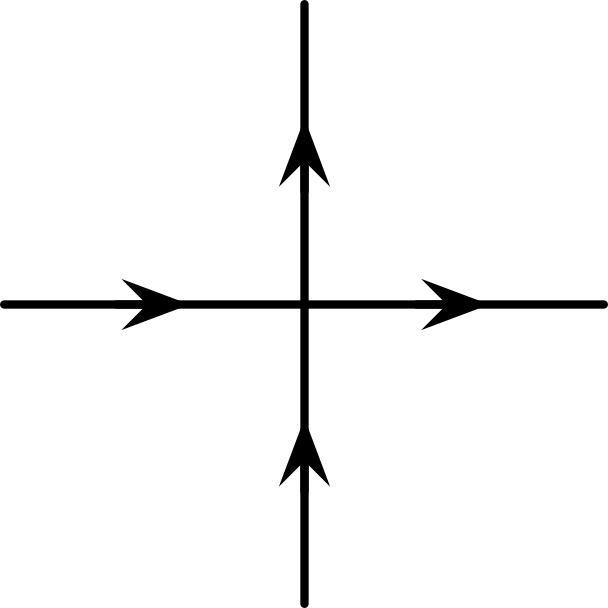}\caption{$1$}
\label{fig:orientations_1}
\end{subfigure}
\begin{subfigure}[b]{0.12\linewidth}
\centering\includegraphics[width=\linewidth]{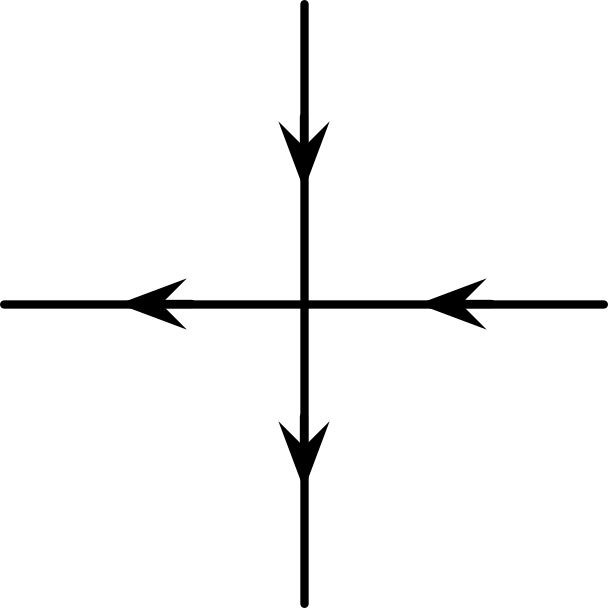}\caption{$2$}
\label{fig:orientations_2}
\end{subfigure}
\begin{subfigure}[b]{0.12\linewidth}
\centering\includegraphics[width=\linewidth]{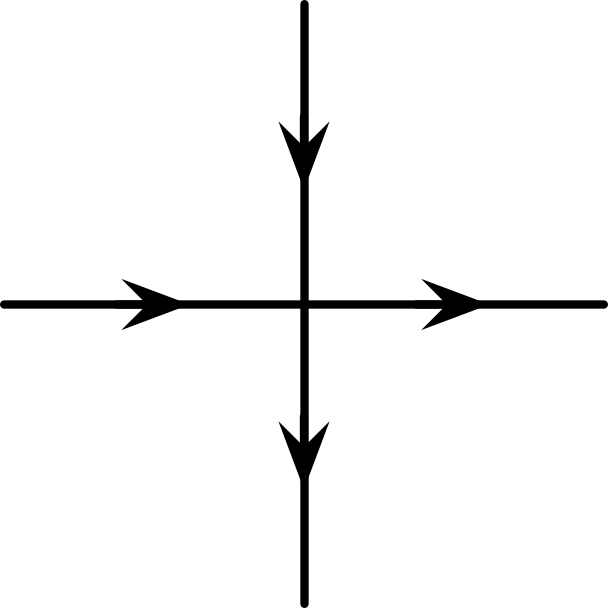}\caption{$3$}
\label{fig:orientations_3}
\end{subfigure}
\begin{subfigure}[b]{0.12\linewidth}
\centering\includegraphics[width=\linewidth]{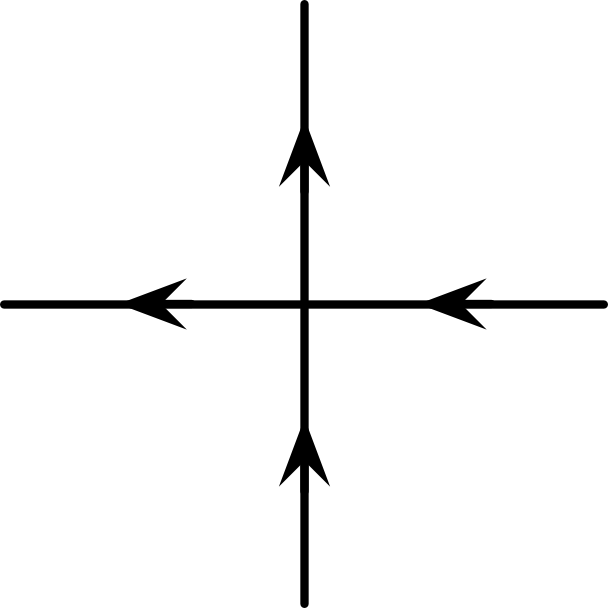}\caption{$4$}
\label{fig:orientations_4}
\end{subfigure}
\begin{subfigure}[b]{0.12\linewidth}
\centering\includegraphics[width=\linewidth]{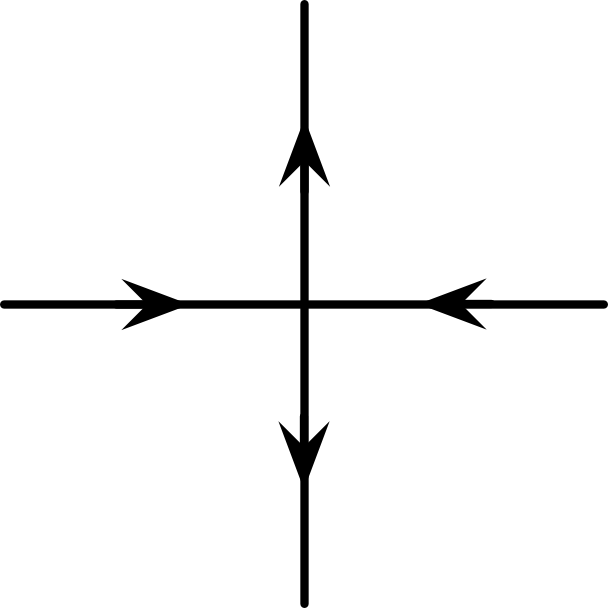}\caption{$5$}
\label{fig:orientations_5}
\end{subfigure}
\begin{subfigure}[b]{0.12\linewidth}
\centering\includegraphics[width=\linewidth]{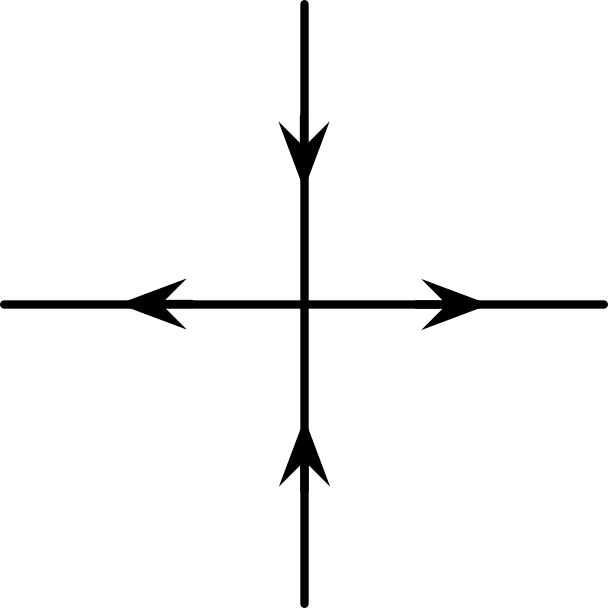}\caption{$6$}
\label{fig:orientations_6}
\end{subfigure}
\begin{subfigure}[b]{0.12\linewidth}
\centering\includegraphics[width=\linewidth]{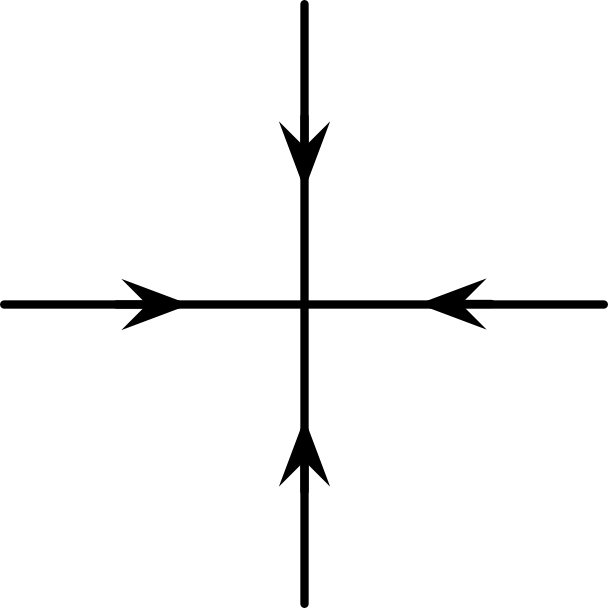}\caption{$7$}
\label{fig:orientations_7}
\end{subfigure}
\begin{subfigure}[b]{0.12\linewidth}
\centering\includegraphics[width=\linewidth]{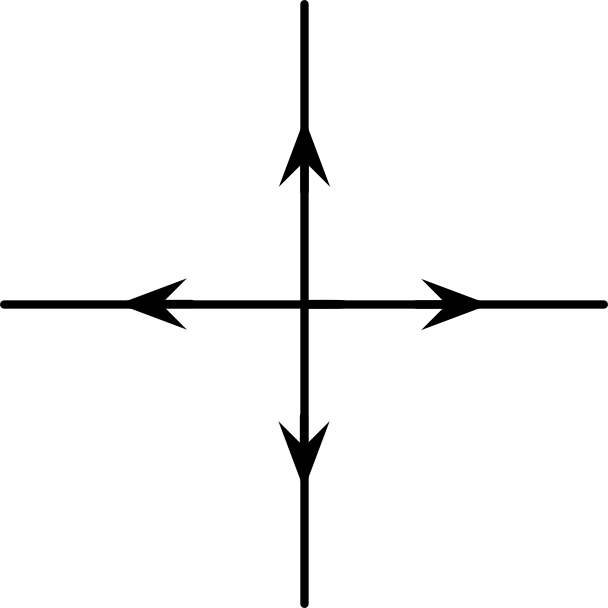}\caption{$8$}
\label{fig:orientations_8}
\end{subfigure}
\caption{Valid configurations of the eight-vertex model.}\label{fig:orientations}
\end{figure}
\renewcommand{\thesubfigure}{\alph{subfigure}}
\captionsetup[subfigure]{labelformat=parens}

Classically, the eight-vertex model
is defined by statistical physicists on a square lattice region where each vertex of the lattice is connected by an edge to four nearest neighbors.
In general, the eight configurations 1 to 8 in \figref{fig:orientations} are associated with eight possible weights $w_1, \ldots, w_8$. 
By physical considerations, the total weight of a state remains unchanged
if  all arrows are flipped,
assuming there is no external electric field.
In this case we write
$w_1 = w_2 = a$, $w_3 = w_4= b$, $w_5 = w_6 = c$, and $w_7 = w_8 = d$.
This complementary invariance is known as \emph{arrow reversal symmetry} or \emph{zero field assumption}.

Even in the zero-field setting, this model is already enormously expressive:
its special case when $d=0$, the zero-field six-vertex model, has sub-models such as the ice ($a = b = c$), KDP, and Rys $F$ models; on the square lattice, some other important models such as the dimer and zero-field Ising models can be reduced to it~\cite{BAXTER1972193}. 
After the eight-vertex model was introduced in 1970 by Sutherland~\cite{doi:10.1063/1.1665111}, and Fan and Wu~\cite{PhysRevB.2.723}, Baxter~\cite{PhysRevLett.26.832, BAXTER1972193} achieved a good understanding of the zero-field case in the thermodynamic limit on the square lattice (in physics it is called ``exactly solved'').

In this paper, we assume the arrow reversal symmetry
and further assume that $a, b, c, d \ge 0$, as is the case in \emph{classical} physics.
Given a  4-regular graph $G$, we label
four incident edges of each vertex
from 1 to 4.
The \emph{partition function} of the eight-vertex model with parameters
$(a, b, c, d)$ on  $G$
is defined as
\begin{equation}\label{Z-defn}
Z_{\textup{\textsc{EightVertex}}}(G; a, b, c, d) = \sum_{\tau \in \mathcal{O}_{\bf e}(G)}a^{n_1 + n_2}b^{n_3 + n_4}c^{n_5 + n_6}d^{n_7 + n_8},
\end{equation}
where $\mathcal{O}_{\bf e}(G)$ is the set of all even orientations of $G$,
and $n_i$ is the number of vertices in type  $i$  in $G$ ($1 \le i \le 8$,
locally depicted as in     
Figure~\ref{fig:orientations}) 
under an even orientation $\tau \in \mathcal{O}_{\bf e}(G)$.

In terms of the exact computational complexity, a complexity dichotomy is given for the eight-vertex model on 4-regular graphs for all eight parameters~\cite{DBLP:journals/corr/CaiF17}.
This is studied in the context of a classification program for the complexity of counting problems~\cite{cai_chen_2017}, where the eight-vertex model serves as important basic cases for Holant problems defined by not necessarily symmetric constraint functions.
It is shown that
every setting is either P-time computable (and some are surprising) or \#P-hard.  
However, most cases for P-time tractability 
are due to nontrivial cancellations.
In our setting where $a, b, c, d$ are nonnegative, the problem of computing the partition function of the eight-vertex model is \#P-hard unless: (1) $a = b = c = d$ (this is equivalent to the unweighted case); (2) at least three of $a, b, c, d$ are zero; or (3) two of $a, b, c, d$ are zero and the other two are equal.
In addition, on planar graphs it is also P-time computable for parameter settings $(a, b, c, d)$ with $a^2 + b^2 = c^2 + d^2$, using the \emph{FKT algorithm}.

Since exact computation is hard in most cases,
one natural question is what is the approximate complexity of counting and sampling of the eight-vertex model.
To our best knowledge, prior to \cite{DBLP:journals/corr/abs-1811-03126}, there is only one
previous result in this regard due to Greenberg and Randall~\cite{Greenberg2010}.
They showed that on square lattice regions
 a specific Markov chain (which flips the orientations of 
all four edges along a uniformly picked face at each step) is torpidly mixing
 when $d$ is 
large. It means that when sinks and sources have
large weights, this particular chain cannot be used to approximately sample eight-vertex configurations on the square lattice according to the Gibbs measure.
Recently, similar torpid mixing results have been achieved for the six-vertex model on the square lattice~\cite{liu:LIPIcs:2018:9456}.

\cite{DBLP:journals/corr/abs-1811-03126} gave the first classification results for the approximate complexity of
the eight-vertex model on general and planar 4-regular graphs,
and they conform to phase transition in physics.
This is an extension to the work on the approximability of the six-vertex model~\cite{doi:10.1137/1.9781611975482.136}.
In order to state the results in \cite{DBLP:journals/corr/abs-1811-03126} and in this work, we adopt the following notations assuming $a, b, c, d \ge 0$.
\vspace{-2mm}
\setlist[itemize]{itemsep=-1mm}
\begin{itemize}
\item
\texttt{DO} = $\{(a,b,c,d) \; | \;
a \le b + c + d, ~~b \le a + c + d,~~ c \le a + b + d,~~d \le a + b + c\}$;
\item
\texttt{$d$-SUM} = $\{(a,b,c,d) \; | \;
a + d \le b + c, ~~b + d \le a + c,~~ c + d \le a + b\}$;
\item
\texttt{SQ-SUM} = $\{(a,b,c,d) \; | \;
a^2 \le b^2 + c^2 + d^2, ~~b^2 \le a^2 + c^2 + d^2,~~ c^2 \le a^2 + b^2 + d^2,~~d^2 \le a^2 + b^2 + c^2\}$.
\end{itemize}
\begin{remark}
\texttt{$d$-SUM} $\subset \texttt{DO}$, \texttt{SQ-SUM} $\subset \texttt{DO}$.
\end{remark}

Physicists have shown an \emph{order-disorder phase transition} for the eight-vertex model on the square lattice between parameter settings outside \texttt{DO} and those inside \texttt{DO} (see Baxter's book~\cite{Baxter:book} for more details).
In \cite{DBLP:journals/corr/abs-1811-03126}, it was shown that: (1) approximating the partition function of the eight-vertex model on general 4-regular graphs outside \texttt{DO} is NP-hard, (2) there is an FPRAS\footnote{Suppose $f: \Sigma^* \rightarrow \mathbb{R}$ is a function mapping problem instances to real numbers. A \textit{fully polynomial randomized approximation scheme (FPRAS)} \cite{Karp:1983:MAE:1382437.1382804} for a problem is a randomized algorithm that takes as input an instance $x$ and $\varepsilon > 0$, running in time polynomial in $n$ (the input length) and $\varepsilon^{-1}$, and outputs a number $Y$ (a random variable) such that
\(\operatorname{Pr}\left[(1 - \varepsilon)f(x) \le Y \le (1 + \varepsilon)f(x)\right] \ge \frac{3}{4}.\)} for general 4-regular graphs in the region \texttt{$d$-SUM} $\bigcap$ \texttt{SQ-SUM}, and (3) there is an FPRAS for planar 4-regular graphs in the extra region $\{(a,b,c,d) \; | \;
a + d \le b + c, ~~b + d \le a + c,~~ c + d \ge a + b\}$ $\bigcap$ \texttt{SQ-SUM}.

\begin{figure}[h!]
\centering\includegraphics[width=0.7\linewidth]{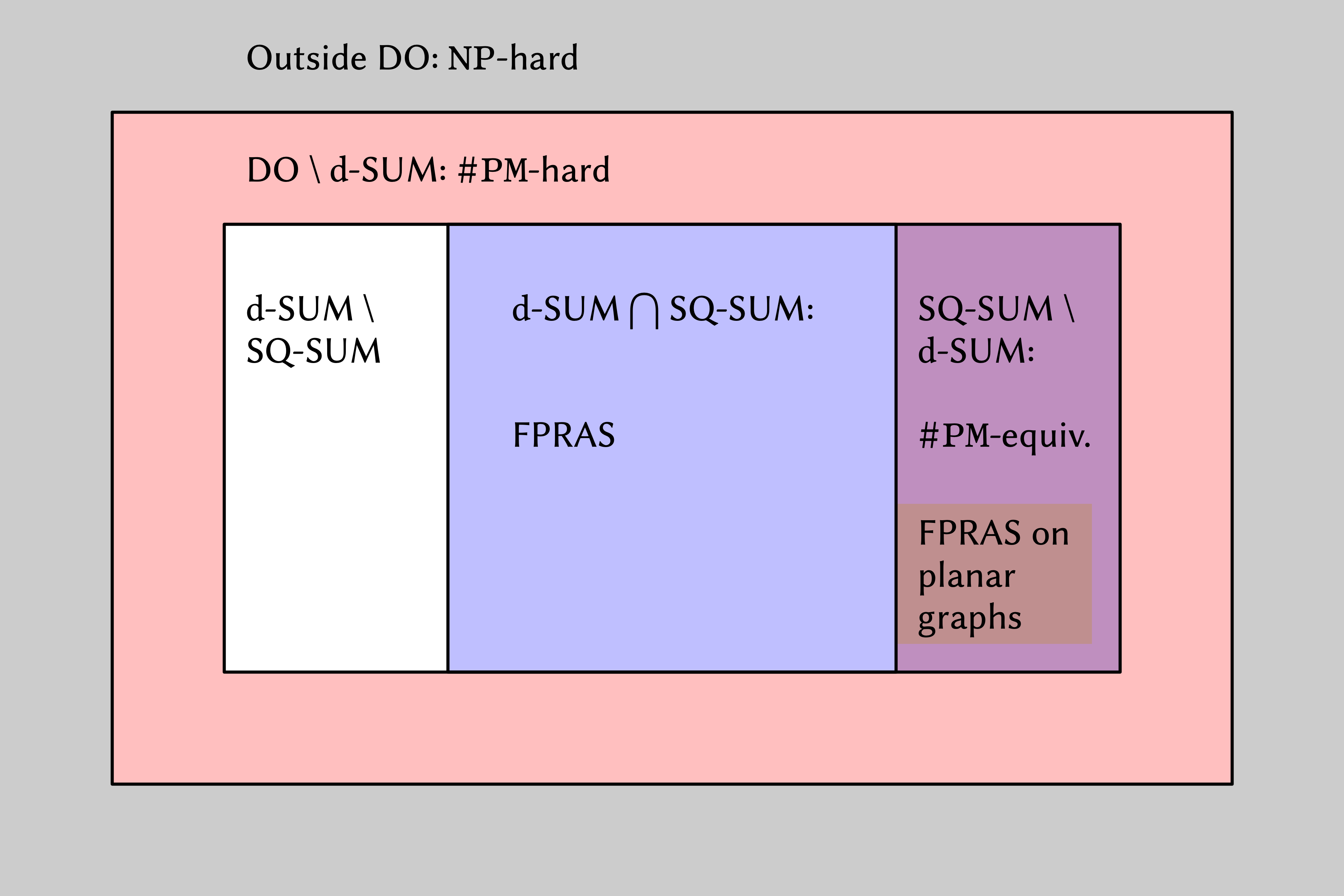}
\caption{A Venn diagram of the approximation complexity of the eight-vertex model.}
\label{fig:complexity_map}
\end{figure}

In this paper we make further progress in the classification program of the approximate complexity of the eight-vertex model on 4-regular graphs in terms of the parameters (see \figref{fig:complexity_map}). For the first time, the complexity of approximating the partition function of the eight-vertex model (\textsc{\#EightVertex($a, b, c, d$)}) is related to that of approximately counting perfect matchings (\textsc{\#PerfectMatchings}).

\begin{theorem}\label{thm:pm-hard}
For any four positive numbers $a,b,c,d >0$ such that
$(a, b, c, d) \not\in \textup{\texttt{$d$-SUM}}$,
the problem \textup{\textsc{\#EightVertex($a, b, c, d$)}} is at least
as hard to approximate as counting perfect matchings:
\[\mbox{
\textup{\textsc{\#PerfectMatchings}} $\le_\textup{AP}$ \textup{\textsc{\#EightVertex($a, b, c, d$)}}.}
\]
\end{theorem}
\begin{remark}
The theorem is stated for the case where all four parameters are positive.
The same proof also works for the case when there is exactly one zero
among the nonnegative values $\{a,b,c\}$. A complete account for
four nonnegative values $\{a,b,c,d\}$ is given in the \appref{app:complexity_table}.
\end{remark}

The proof of \thmref{thm:pm-hard} is in \secref{sec:hardness}.
Our proof for the hardness result has several ingredients:
\vspace{-2mm}
\setlist[enumerate]{itemsep=-1mm}
\begin{enumerate}[(1)]
\item
We express the eight-vertex model on a 4-regular graph $G$ as an edge-2-coloring problem on $G$ using Valiant's \emph{holographic transformation}~\cite{Valiant:2008:HA:1350684.1350697}.
\item
We show that some special cases in this edge-2-coloring problem on $G$ is equivalent to the zero-field Ising model on its \emph{crossing-circuit graph} $\tilde{G}$. Thus known \#\textsc{PerfectMatchings}-equivalence result for the Ising model~\cite[Lemma 7]{GOLDBERG2008908} directly transfers to the special cases under certain parameter settings.
\item
We further show that for any parameter setting outside \texttt{$d$-SUM}, approximating the partition function of the eight-vertex model is at least as hard as the \#\textsc{PerfectMatchings}-equivalent special cases of the edge-2-coloring problem via \emph{approximation-preserving reductions} (introduced in \cite{Dyer2004}). 
\end{enumerate}

\begin{theorem}\label{thm:pm-easy}
For any $(a, b, c, d) \in \textup{\texttt{SQ-SUM}}$,
\[\mbox{
\textup{\textsc{\#EightVertex($a, b, c, d$)}} $\le_\textup{AP}$ \textup{\textsc{\#PerfectMatchings}}.}
\]
\end{theorem}

The proof of \thmref{thm:pm-easy} is in \secref{sec:easiness}.
To prove the easiness result, we again express the eight-vertex model in the Holant framework (see \secref{sec:prelim}) and show that the constraint functions of the eight-vertex model in \texttt{SQ-SUM} can be implemented by constant-size \emph{matchgates} with nonnegatively weighted edges (\defnref{defn:matchgate}).
We note that allowing nonnegative edge-weights does not add more computational power~\cite[Proposition 5]{DBLP:journals/corr/abs-1301-2880}.
The crucial ingredient of our proof is a \emph{geometric} lemma (\lemref{lem:pm-easy_geometry}) in 3-dimensional space.

Our result is tight in the sense that no constraint functions of the eight-vertex model with parameter settings outside the region \texttt{SQ-SUM} can be implemented by a matchgate (\lemref{lem:pm-easy_tight}). 
Moreover, the general version of our result also works for the eight-vertex model without the arrow reversal symmetry.
It is open if computing the partition function in $\texttt{DO} \setminus (\texttt{$d$-SUM} \bigcup \texttt{SQ-SUM})$ is \textsc{\#PerfectMatchings}-equivalent or not.

As part of this work, we give a complete characterization of the constraint functions that can be expressed by 4-ary matchgates in \thmref{thm:matchgate}. This solves an important question that has been open for several years~\cite{DBLP:journals/corr/abs-1301-2880, BULATOV201711}.
We believe it is of independent interest.

\begin{corollary}\label{cor:pm-equiv}
For any four positive numbers $a,b,c,d >0$ such that
$(a, b, c, d) \in \textup{\texttt{SQ-SUM}} \setminus \textup{\texttt{$d$-SUM}}$,
\[\mbox{
\textup{\textsc{\#EightVertex($a, b, c, d$)}} $\equiv_\textup{AP}$ \textup{\textsc{\#PerfectMatchings}}.}
\]
\end{corollary}

Note that for the eight-vertex model in the region $\{(a,b,c,d) \; | \;
a + d \le b + c, ~~b + d \le a + c,~~ c + d > a + b\}$ $\bigcap$ \texttt{SQ-SUM}, computing $Z_{\textup{\textsc{EightVertex}}}(a, b, c, d)$ is (1) \#P-complete in exact computation~\cite{DBLP:journals/corr/CaiF17}, (2) \textsc{\#PerfectMatchings}-equivalent in approximate computation on general 4-regular graphs (\corref{cor:pm-equiv}), and (3) admits an FPRAS in approximate computation on planar 4-regular graphs~\cite{DBLP:journals/corr/abs-1811-03126}.
To our best knowledge, these are the first identified problems of this kind.

\section{Preliminaries}\label{sec:prelim}


Given a 4-regular graph $G = (V, E)$,
the \emph{edge-vertex incidence graph}  $G' = (U_E, U_V, E')$
is a bipartite graph where
$(u_e, u_v) \in U_E \times U_V$ is an edge in $E'$ iff 
$e \in E$ in $G$  is incident to $v \in V$.
We model an orientation ($w \rightarrow v$)
 on an edge $e = \left\{w, v\right\} \in E$
 from $w$ into $v$ in $G$ by assigning
 $1$  to $(u_e, u_w) \in E'$  and  $0$ to $(u_e, u_v) \in E'$
in $G'$.
A configuration of the eight-vertex model on $G$
is  an  \emph{edge 2-coloring} on $G'$,
namely $\sigma: E' \rightarrow \{0, 1\}$,
where for  each  $u_e \in U_E$ its two incident edges are
assigned 01 or 10, and  for  each $u_v \in U_V$ the  sum of
 values $\sum_{i=1}^4 \sigma(e_i) \equiv 0  \pmod 2$,
over the four incident edges of $u_v$.
Thus 
we model the even orientation rule of $G$ on all $v \in V$ by requiring ``two-0-two-1/four-0/four-1'' locally at
each vertex $u_v \in U_V$.

The ``one-0-one-1'' requirement on the two edges incident to a vertex in $U_E$ is a binary {\sc Disequality} constraint, denoted by $(\neq_2)$.
The values of a 4-ary \emph{constraint function} $f$  can be listed in a matrix $M(f) = \left[\begin{smallmatrix} f_{0000} & f_{0010} & f_{0001} & f_{0011} \\ f_{0100} & f_{0110} & f_{0101} & f_{0111} \\ f_{1000} & f_{1010} & f_{1001} & f_{1011} \\ f_{1100} & f_{1110} & f_{1101} & f_{1111}\end{smallmatrix}\right]$,
called the \emph{constraint matrix} of $f$. For the eight-vertex model 
satisfying the even orientation rule and arrow reversal symmetry, the constraint function $f$ at every vertex $v \in U_V$ in $G'$ 
has the form $M(f) = \left[\begin{smallmatrix} d & 0 & 0 & a \\ 0 & b & c & 0 \\ 0 & c & b & 0 \\ a & 0 & 0 & d \end{smallmatrix}\right]$, if we locally index the left, down, right, and up edges incident to $v$ by 1, 2, 3, and 4, respectively according to \figref{fig:orientations}.
Thus computing the partition function $Z_{\textup{\textsc{EightVertex}}}(G; a, b, c, d)$ is equivalent to evaluating 
\[Z'(G'; f) := \sum_{\sigma:E'\rightarrow\left\{0,1\right\}}\prod_{u\in U_E}(\neq_2)\left(\sigma |_{E'(u)}\right) \prod_{u\in U_V}f\left(\sigma |_{E'(u)}\right).\]
where $E'(u)$ denotes the incident edges of $u \in U_E \cup U_V$.
In fact, in this way we express the partition function of the eight-vertex model as the Holant sum in the framework for Holant problems:
\[Z_{\textup{\textsc{EightVertex}}}(G; a, b, c, d) = \textup{Holant}\left(G'; \neq_2 |\ f\right)\]
where we use $\textup{Holant}(H; g\ |\ f)$ to denote the Holant sum 
$\sum_{\sigma:E\rightarrow\left\{0,1\right\}}\prod_{u\in U}g\left(\sigma |_{E(u)}\right) \prod_{u\in V}f\left(\sigma |_{E(u)}\right)$
on a bipartite graph $H = (U, V, E)$ for the Holant problem $\textup{Holant}(g\ |\ f)$.
Each vertex in $U$ (or $V$) is assigned the constraint function $g$ (or $f$, respectively).
The constraint function $g$ is considered as a row vector (or covariant tensor), whereas the signature $f$ is considered as a column vector (or contravariant tensor).
(See \cite{cai_chen_2017} for more on Holant problems.)
The following proposition says that an invertible holographic transformation does not change the complexity of the Holant problem in the bipartite setting.
\begin{proposition}[\cite{Valiant:2008:HA:1350684.1350697}]\label{prop:holo_trans}
Suppose $T \in \mathbb{C}^2$ is an invertible matrix. Let $d_1 = \operatorname{arity}(g)$ and $d_2 = \operatorname{arity}(f)$. Define $g' = g \left(T^{-1}\right)^{\otimes d_1}$ and $f' = T^{\otimes d_2} f$.
Then for any bipartite graph $H$, $\textup{Holant}(H; g\ |\ f) = \textup{Holant}(H; g'\ |\ f')$.
\end{proposition}


\section{\#\textsc{PerfectMatchings}-hardness}\label{sec:hardness}

Our proof strategy for \thmref{thm:pm-hard} is as follows.
In \lemref{lem:pm-hard_holant}, 
we express the eight-vertex model on a 4-regular graph $G$  
as a Holant problem; this is an equivalent form of the orientation
problem expressed as an edge-2-coloring problem on $G$, and is achieved using 
a holographic transformation.
In \lemref{lem:pm-hard_base}, we establish the equivalence between some special cases of this edge-2-coloring problem and the zero-field Ising model. Thus a known result for the Ising model (\propref{thm:ising}) indicates the \#\textsc{PerfectMatchings}-equivalence of the special cases under certain parameter settings (\corref{cor:pm-hard_base}).
It follows from \lemref{lem:pm-hard_general} that for any $(a, b, c, d)$ with $a + d > b + c$ (and symmetrically  $b + d > a + c$ or $c + d > a + b$), approximately computing the partition function is at least as hard as the \#\textsc{PerfectMatchings}-equivalent special cases under approximation-preserving reductions.

\begin{lemma}\label{lem:pm-hard_holant}
$2^{|V(G)|} \cdot Z_{\textup{\textsc{EightVertex}}}(G; a, b, c, d) = \textup{Holant}\left(G; \left[\begin{smallmatrix} a + b + c + d & 0 & 0 & - a + b + c - d \\ 0 & a - b + c - d & a +b - c - d & 0 \\ 0 & a +b - c - d & a - b + c - d & 0 \\ - a + b + c - d & 0 & 0 & a + b + c + d \end{smallmatrix}\right]\right)$.
\end{lemma}
\begin{proof}
Using the binary disequality function $(\not=_2)$ for
the orientation of any edge, we can express  the partition function of the eight-vertex model $G$ as a Holant problem on its edge-vertex incidence graph $G'$,
\[Z_{\textup{\textsc{EightVertex}}}(G; a, b, c, d) = \textup{Holant}\left(G'; \neq_2 |\ f\right),\]
where $f$ is the 4-ary constraint function with $M(f) = \left[\begin{smallmatrix} d & 0 & 0 & a \\ 0 & b & c & 0 \\ 0 & c & b & 0 \\ a & 0 & 0 & d \end{smallmatrix}\right]$.
According to \propref{prop:holo_trans}, under a $Z$-transformation where $Z = \frac{1}{\sqrt{2}}\left[\begin{smallmatrix} 1 & 1 \\ i & -i \end{smallmatrix}\right]$, we have
\begin{align*}
\textup{Holant}\left(G'; \neq_2 |\ f\right)
& = \textup{Holant}\left(G'; \neq_2 \cdot \left( Z^{-1} \right)^{\otimes 2} |\ Z^{\otimes 4} \cdot f\right) \\
& = \textup{Holant}\left(G'; =_2 |\ Z^{\otimes 4} f\right)\\
& = \textup{Holant}\left(G; \ Z^{\otimes 4} f\right),
\end{align*}
and a direct calculation shows that $M(Z^{\otimes 4} f) = \frac{1}{2}
\left[\begin{smallmatrix} a + b + c + d & 0 & 0 & - a + b + c - d \\ 0 & a - b + c - d & a +b - c - d & 0 \\ 0 & a +b - c - d & a - b + c - d & 0 \\ - a + b + c - d & 0 & 0 & a + b + c + d \end{smallmatrix}\right]$.
\end{proof}

\begin{algorithm}[H]
\Problem{\textsc{Ising}$(\beta)$.}
\Instance{Graph $G = (V, E)$.}
\KwOut{$Z_\textup{\textsc{Ising}}(G; \beta) := \sum\limits_{\sigma: V \rightarrow \{0, 1\}} \beta^{\textup{mono}(\sigma)}$, where $\textup{mono}(\sigma)$ denotes the number of
edges $\{u, v\}$ such that  $\sigma(u) = \sigma(v)$.}
\end{algorithm}

\begin{proposition}[{\cite[Lemma 7]{GOLDBERG2008908}}]\label{thm:ising}
Suppose $\beta < -1$. Then $\textup{\textsc{\#PerfectMatchings}} \equiv_\textup{AP} \textup{\textsc{Ising}}(\beta)$.
\end{proposition}

\begin{lemma}\label{lem:pm-hard_base}
The Ising problem $\textup{\textsc{Ising}}\left(\frac{w}{x}\right)$
is equivalent to the Holant problem $\textup{Holant}\left(\left[\begin{smallmatrix} w & 0 & 0 & x \\ 0 & 0 & 0 & 0 \\ 0 & 0 & 0 & 0 \\ x & 0 & 0 & w \end{smallmatrix}\right]\right)$.
In particular,
$\textup{\textsc{Ising}}\left(\frac{w}{x}\right) \equiv_\textup{AP} \textup{Holant}\left(\left[\begin{smallmatrix} w & 0 & 0 & x \\ 0 & 0 & 0 & 0 \\ 0 & 0 & 0 & 0 \\ x & 0 & 0 & w \end{smallmatrix}\right]\right)$.
\end{lemma}
\begin{remark}
A non-homogenized form of the Ising model is
$\tilde{Z}_\textup{\textsc{Ising}}(G; x, w) := 
\sum\limits_{\sigma: V \rightarrow \{0, 1\}} w^{\textup{mono}(\sigma)}
x^{|E| - \textup{mono}(\sigma)}$. If $x \not =0$ then
$\tilde{Z}_\textup{\textsc{Ising}}(G; x, w) = x^{|E|} Z_\textup{\textsc{Ising}}(G; \frac{w}{x})$.
If $x = 0$ then in $\tilde{Z}_\textup{\textsc{Ising}}$
all vertices in each component must take the same assignment (all $0$ or all $1$).   
In this case both $\tilde{Z}_\textup{\textsc{Ising}}(G; x, w)$
and the Holant problem in Lemma~\ref{lem:pm-hard_base}
 are trivially solvable in polynomial time.
\end{remark}
\begin{proof}
For the problem $\textup{Holant}\left(\left[\begin{smallmatrix} w & 0 & 0 & x \\ 0 & y & z & 0 \\ 0 & z & y & 0 \\ x & 0 & 0 & w \end{smallmatrix}\right]\right)$, the roles of $x, y, z$ are interchangeable by relabeling the edges.
For example, if the constraint function $f(x_1,x_2,x_3,x_4)$ has
the constraint matrix
$\left[\begin{smallmatrix} w & 0 & 0 & x \\ 0 & y & z & 0 \\ 0 & z & y & 0 \\ x & 0 & 0 & w \end{smallmatrix}\right]$, 
then the constraint function $f(x_1,x_3,x_2,x_4)$ has the constraint matrix
$\left[\begin{smallmatrix} w & 0 & 0 & z \\ 0 & y & x & 0 \\ 0 & x & y & 0 \\ z & 0 & 0 & w \end{smallmatrix}\right]$.
It follows that 
\[
\textup{Holant}\left(\left[\begin{smallmatrix} w & 0 & 0 & x \\ 0 & 0 & 0 & 0 \\ 0 & 0 & 0 & 0 \\ x & 0 & 0 & w \end{smallmatrix}\right]\right)
~~~\mbox{and} ~~~
\textup{Holant}\left(\left[\begin{smallmatrix} w & 0 & 0 & 0 \\ 0 & 0 & x & 0 \\ 0 & x & 0 & 0 \\ 0 & 0 & 0 & w \end{smallmatrix}\right]\right)
\] are exactly the same problem. So to prove the lemma 
it suffices to prove 
the equivalence of
\[
\textup{\textsc{Ising}}\left(\frac{w}{z}\right) 
~~~\mbox{and} ~~~
 \textup{Holant}\left(\left[\begin{smallmatrix} w & 0 & 0 & 0 \\ 0 & 0 & z & 0 \\ 0 & z & 0 & 0 \\ 0 & 0 & 0 & w \end{smallmatrix}\right]\right).
\]

First we show that
\(
\textup{Holant}\left(\left[\begin{smallmatrix} w & 0 & 0 & 0 \\ 0 & 0 & z & 0 \\ 0 & z & 0 & 0 \\ 0 & 0 & 0 & w \end{smallmatrix}\right]\right)\)
can be expressed as
\( \textup{\textsc{Ising}}\left(\frac{w}{z}\right).
\)

Given a 4-regular graph $G = (V, E)$
as an instance of $\textup{Holant}\left(\left[\begin{smallmatrix} w & 0 & 0 & 0 \\ 0 & 0 & z & 0 \\ 0 & z & 0 & 0 \\ 0 & 0 & 0 & w \end{smallmatrix}\right]\right)$, we can partition $E$ into a set $\mathcal{C}$ of \emph{circuits} (in which vertices may repeat but edges cannot) in the following way:
at every vertex $v \in V$, denote the four edges incident to $v$ by $e_1, e_2, e_3, e_4$ in a cyclic order according to the local labeling of the signature
function;
we make $e_1$ and $e_3$ into adjacent edges in a single circuit,
and similarly we make $e_2$ and $e_4$ into  adjacent edges in a single circuit
(note that these may be the same circuit).
We say each circuit in $\mathcal{C}$ is a \emph{crossing circuit} of $G$.
%
For the graph $G$, we define its \emph{crossing-circuit graph} $\tilde{G} = (\mathcal{C}, \tilde{E})$, with possible multiloops and multiedges,  as follows:
its vertex set  $\mathcal{C}$ consists of the crossing circuits;
for every $v \in V$, if 
circuits $C_1$ and $C_2$  intersect at $v$, then there is an edge 
 $\tilde{e}_v \in \tilde{E}$ labeled by $v$.
Note that it is possible that  $C_1= C_2$, and for such a
self-intersectison point the edge $\tilde{e}_v$ is a loop.
Each $C \in \mathcal{C}$ may have multiple loops,
and for distinct circuits $C_1$ and $C_2$ there may be multiple edges
between them. The edge set $\tilde{E}$ of $\tilde{G}$
is in 1-1 correspondence with $V$ of $G$.

Observe that the problem $\textup{Holant}\left(\left[\begin{smallmatrix} w & 0 & 0 & 0 \\ 0 & 0 & z & 0 \\ 0 & z & 0 & 0 \\ 0 & 0 & 0 & w \end{smallmatrix}\right]\right)$ requires that every valid configuration $\sigma$ (that contributes a non-zero term) obeys the following rule at each vertex $v$:
\begin{itemize}
\item
Assuming $e_1, e_2, e_3, e_4$ are the four edges incident to $v$ in cyclic 
order,
then $\sigma(e_1) = \sigma(e_3)$ (denoted by $b_1$) and $\sigma(e_2) = \sigma(e_4)$ (denoted by $b_2$). That is to say, all edges in a crossing circuit must have the same assignment (either all $0$ or all $1$).
Therefore, the valid configurations $\sigma$ on the edges of $G$ 
are in 1-1 correspondence with $0,1$-assignments $\sigma'$ on the vertices of $\tilde{G}$.
\item
Under $\sigma$, the local weight on $v$ is $w$ if $b_1 = b_2$ and is $z$ otherwise.
Suppose crossing circuits  $C_1$ and $C_2$  intersect at $v$ (they could be identical). Then in $\tilde{G}$, $\sigma'$ has local weight $w$ on the edge $\tilde{e}_v$ if $\sigma'(C_1) = \sigma'(C_2)$ and has local weight $z$ otherwise.
\end{itemize}
This means
\[
\textup{Holant}\left(G; \left[\begin{smallmatrix} w & 0 & 0 & 0 \\ 0 & 0 & z & 0 \\ 0 & z & 0 & 0 \\ 0 & 0 & 0 & w \end{smallmatrix}\right]\right) = z^{|V(G)|} \cdot Z_\textup{\textsc{Ising}}\left(\tilde{G}; \frac{w}{z}\right).
\]

Next we show that
\( \textup{\textsc{Ising}}\left(\frac{w}{z}\right)
\)
can be expressed as
\(
\textup{Holant}\left(\left[\begin{smallmatrix} w & 0 & 0 & 0 \\ 0 & 0 & z & 0 \\ 0 & z & 0 & 0 \\ 0 & 0 & 0 & w \end{smallmatrix}\right]\right).\)
Note that every graph $G = (V, E)$ (without
isolated vertices) is the crossing-circuit graph of some 4-regular graph $\bar{G}$.
To define $\bar{G}$ from $G$, one only needs to
do the following: (1) transform each vertex $v \in V$ into a closed cycle $C_v$;
(2) for each loop at $v \in V$, make a self-intersection on $C_v$;
and (3) for each non-loop edge $\{u, v\} \in E$ ($u$ and $v$ are two distinct vertices), make $C_u$ and $C_v$ intersect in a ``crossing'' way at a vertex in $\bar{G}$ (by first creating a vertex $p$ on $C_u$ and another vertex $p'$ on $C_v$, then merging $p$ and $p'$ with local labeling $1,3$ on $C_u$ and $2, 4$ on $C_v$).
Then the above proof holds for the reverse direction.
\end{proof}

\begin{corollary}\label{cor:pm-hard_base}
Suppose $x \neq 0$ and $\frac{w}{x} < -1$. Then $\textup{\textsc{\#PerfectMatchings}} \equiv_\textup{AP} \textup{Holant}\left(\left[\begin{smallmatrix} w & 0 & 0 & x \\ 0 & 0 & 0 & 0 \\ 0 & 0 & 0 & 0 \\ x & 0 & 0 & w \end{smallmatrix}\right]\right).$
\end{corollary}

\begin{lemma}\label{lem:pm-hard_general}
Suppose $d > 0$ and at most one of $a, b, c$ is zero.
Then $\textup{Holant}\left(\left[\begin{smallmatrix} a+b+c+d & 0 & 0 & -a+b+c-d \\ 0 & 0 & 0 & 0 \\ 0 & 0 & 0 & 0 \\ -a+b+c-d & 0 & 0 & a+b+c+d \end{smallmatrix}\right]\right) \le_\textup{AP} \textup{\textsc{\#EightVertex}}(a, b, c, d)$.
\end{lemma}
\begin{proof}
According to \lemref{lem:pm-hard_holant}, we only need to prove
\begin{equation}\label{eqn:result_1}
\textup{Holant}\left(\left[\begin{smallmatrix} a+b+c+d & 0 & 0 & -a+b+c-d \\ 0 & 0 & 0 & 0 \\ 0 & 0 & 0 & 0 \\ -a+b+c-d & 0 & 0 & a+b+c+d \end{smallmatrix}\right]\right) \le_\textup{AP}
\textup{Holant}\left(\left[\begin{smallmatrix} a + b + c + d & 0 & 0 & - a + b + c - d \\ 0 & a - b + c - d & a +b - c - d & 0 \\ 0 & a +b - c - d & a - b + c - d & 0 \\ - a + b + c - d & 0 & 0 & a + b + c + d \end{smallmatrix}\right]\right).
\end{equation}

\figref{fig:2-chain} describes a simple gadget construction
in the Holant framework.
In \figref{fig:2-chain} if we place the constraint function $f_1$ with $M_{x_ix_j,x_lx_k}(f_1)$ and $f_2$ with $M_{x_px_q,x_sx_r}(f_2)$ at 
the two degree 4 vertices, then the constraint function $f_3$ of the 
 4-ary construction is $M(f_3) = M_{x_ix_j,x_sx_r}(f_3)
= M_{x_ix_j,x_lx_k}(f_1) \cdot M_{x_px_q,x_sx_r}(f_2)$.
Since
\[\left[\begin{smallmatrix} a+b+c+d & 0 & 0 & -a+b+c-d \\ 0 & 0 & 0 & 0 \\ 0 & 0 & 0 & 0 \\ -a+b+c-d & 0 & 0 & a+b+c+d \end{smallmatrix}\right] = \left[\begin{smallmatrix} a + b + c + d & 0 & 0 & - a + b + c - d \\ 0 & a - b + c - d & a +b - c - d & 0 \\ 0 & a +b - c - d & a - b + c - d & 0 \\ - a + b + c - d & 0 & 0 & a + b + c + d \end{smallmatrix}\right] \cdot \left[\begin{smallmatrix} 1 & 0 & 0 & 0 \\ 0 & 0 & 0 & 0 \\ 0 & 0 & 0 & 0 \\ 0 & 0 & 0 & 1 \end{smallmatrix}\right],\]
we know that the Holant problem below on the left can be expressed 
by the Holant problem on the  right, and therefore
\[
\textup{Holant}\left(\left[\begin{smallmatrix} a+b+c+d & 0 & 0 & -a+b+c-d \\ 0 & 0 & 0 & 0 \\ 0 & 0 & 0 & 0 \\ -a+b+c-d & 0 & 0 & a+b+c+d \end{smallmatrix}\right]\right) \le_\textup{AP}
\textup{Holant}\left(\left[\begin{smallmatrix} a + b + c + d & 0 & 0 & - a + b + c - d \\ 0 & a - b + c - d & a +b - c - d & 0 \\ 0 & a +b - c - d & a - b + c - d & 0 \\ - a + b + c - d & 0 & 0 & a + b + c + d \end{smallmatrix}\right], \left[\begin{smallmatrix} 1 & 0 & 0 & 0 \\ 0 & 0 & 0 & 0 \\ 0 & 0 & 0 & 0 \\ 0 & 0 & 0 & 1 \end{smallmatrix}\right]\right).
\]
Notice that $\left[\begin{smallmatrix} 1 & 0 & 0 & 0 \\ 0 & 0 & 0 & 0 \\ 0 & 0 & 0 & 0 \\ 0 & 0 & 0 & 1 \end{smallmatrix}\right]$ is the signature matrix of the
arity 4 equality function $(=_4)$.
Therefore, (\ref{eqn:result_1}) is true if we can show
\[
\textup{Holant}\left(\left[\begin{smallmatrix} a + b + c + d & 0 & 0 & - a + b + c - d \\ 0 & a - b + c - d & a +b - c - d & 0 \\ 0 & a +b - c - d & a - b + c - d & 0 \\ - a + b + c - d & 0 & 0 & a + b + c + d \end{smallmatrix}\right], \left[\begin{smallmatrix} 1 & 0 & 0 & 0 \\ 0 & 0 & 0 & 0 \\ 0 & 0 & 0 & 0 \\ 0 & 0 & 0 & 1 \end{smallmatrix}\right]\right)
\le_\textup{AP}
\textup{Holant}\left(\left[\begin{smallmatrix} a + b + c + d & 0 & 0 & - a + b + c - d \\ 0 & a - b + c - d & a +b - c - d & 0 \\ 0 & a +b - c - d & a - b + c - d & 0 \\ - a + b + c - d & 0 & 0 & a + b + c + d \end{smallmatrix}\right]\right),
\]
which is equivalent to the following in the orientation view according to the proof of \lemref{lem:pm-hard_holant}
\begin{equation}\label{eqn:result_2}
\textup{Holant}\left(\neq_2 | \left[\begin{smallmatrix} d & 0 & 0 & a \\ 0 & b & c & 0 \\ 0 & c & b & 0 \\ a & 0 & 0 & d \end{smallmatrix}\right], \left[\begin{smallmatrix} 1 & 0 & 0 & 1 \\ 0 & 1 & 1 & 0 \\ 0 & 1 & 1 & 0 \\ 1 & 0 & 0 & 1 \end{smallmatrix}\right]\right)
\le_\textup{AP}
\textup{Holant}\left(\neq_2 | \left[\begin{smallmatrix} d & 0 & 0 & a \\ 0 & b & c & 0 \\ 0 & c & b & 0 \\ a & 0 & 0 & d \end{smallmatrix}\right]\right),
\end{equation}
by the holographic transformation
\[4 (Z^{-1})^{\otimes 4} ( (=_4)) =
\begin{bmatrix} 1 & - i \\ 1 & i \end{bmatrix}^{\otimes 4}
\left[ \left[\begin{matrix} 1 \\ 0 \end{matrix}\right]^{\otimes 4}
+
\left[\begin{matrix} 0 \\ 1 \end{matrix}\right]^{\otimes 4}
\right]
=
 \left[\begin{matrix} 1 \\ 1 \end{matrix}\right]^{\otimes 4}
+
\left[\begin{matrix} -i \\ i \end{matrix}\right]^{\otimes 4}
=
2 [1,0,1,0,1],\]
a constant multiple of the even parity function,
which has the signature matrix
$\left[\begin{smallmatrix} 1 & 0 & 0 & 1 \\ 0 & 1 & 1 & 0 \\ 0 & 1 & 1 & 0 \\ 1 & 0 & 0 & 1 \end{smallmatrix}\right]$.

\begin{figure}[h!]
\centering
\includegraphics[width=0.4\linewidth]{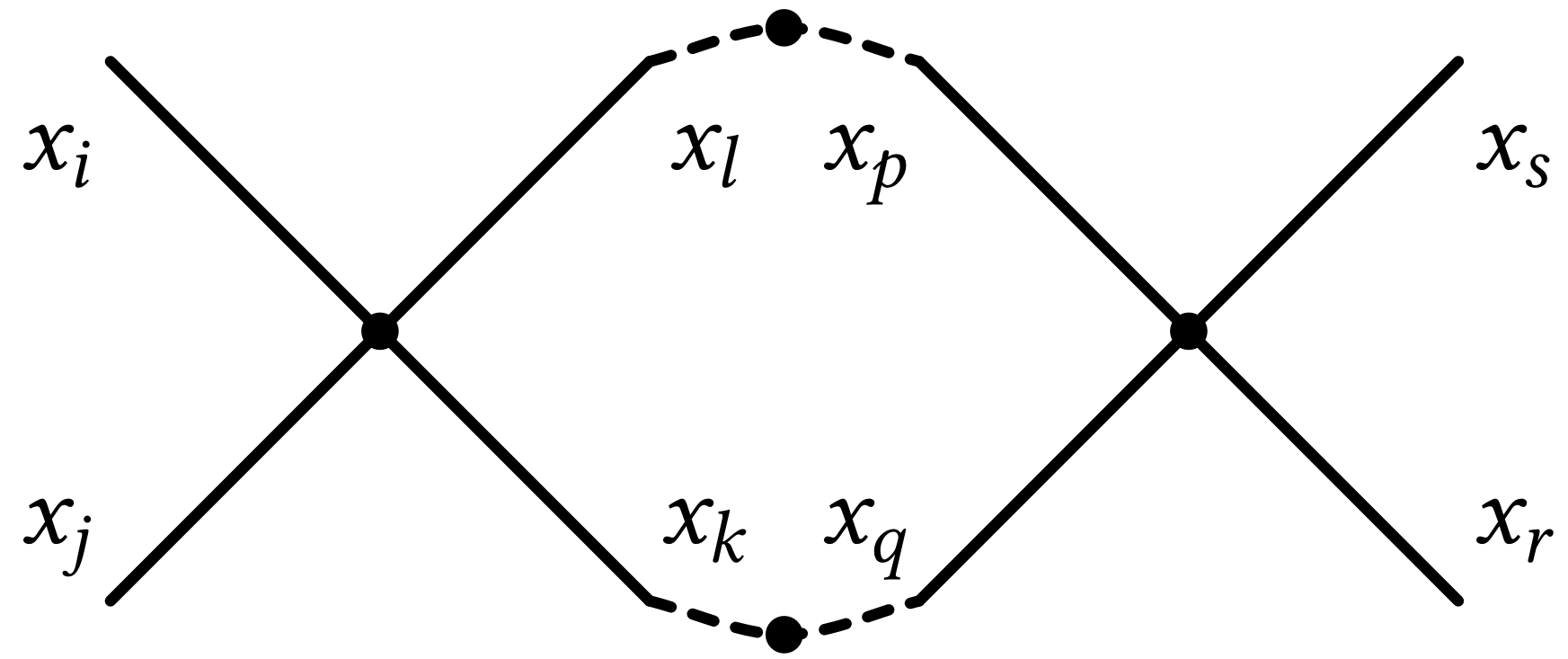}
\caption{}
\label{fig:2-chain}
\end{figure}

Next we show how to get (\ref{eqn:result_2}).
Given  the constraint function $f$ with
matrix
$\left[\begin{smallmatrix} d & 0 & 0 & a \\ 0 & b & c & 0 \\ 0 & c & b & 0 \\ a & 0 & 0 & d \end{smallmatrix}\right]$ 
in \textup{\textsc{\#EightVertex}}($a, b, c, d$),
we construct a 4-ary signature $\check{f}$ with constraint matrix 
$\left[\begin{smallmatrix} \check{d} & 0 & 0 & \check{a} \\ 0 & \check{b} & \check{c} & 0 \\ 0 & \check{c} & \check{b} & 0 \\ \check{a} & 0 & 0 & \check{d} \end{smallmatrix}\right]$
using a polynomial number of vertices and edges
such that $\check{a}$, $\check{b}$, $\check{c}$, and $\check{d}$ are all
exponentially close to $1$ after
normalization, i.e., to be $2^{-n^C}$ close to 1, for any $C>0$,
  with a construction of $n^{O(1)}$ size in polynomial time.

We assume we start with the following condition:
\begin{equation}\label{eqn:condition}
0 < d \le a \le b \le c. 
\end{equation}
If this is not the case, we can obtain a 4-ary construction that realizes this condition using constantly many vertices.
With some preliminary construction we can further assume
$1 \le d \le a \le b \le c \le \frac{3}{2}d$ initially.  (See \appref{app:star_condition} for details.)
Note that starting with the constraint function $f$ with  matrix $M(f) = \left[\begin{smallmatrix} d & 0 & 0 & a \\ 0 & b & c & 0 \\ 0 & c & b & 0 \\ a & 0 & 0 & d \end{smallmatrix}\right]$, we can 
arbitrarily permute $a, b, c$ by relabeling the edges,
and so we get constraint functions $f_1$ with
$M(f_1) = \left[\begin{smallmatrix} d & 0 & 0 & b \\ 0 & a & c & 0 \\ 0 & c & a & 0 \\ b & 0 & 0 & d \end{smallmatrix}\right]$ and $f_2$ with
$M(f_2) = \left[\begin{smallmatrix} d & 0 & 0 & c \\ 0 & a & b & 0 \\ 0 & b & a & 0 \\ c & 0 & 0 & d \end{smallmatrix}\right]$.
There are two constructions $G_1$ and $G_2$ which we use as basic steps; both
constructions start with
a constraint function $f$ with parameters satisfying (\ref{eqn:condition}).
\begin{enumerate}
\item
$\boldsymbol{G_1}$: connect two vertices with constraint functions $f_1$ and $f_2$ respectively as in \figref{fig:2-chain}.
Since we are in the orientation view, we place the constraint function $(\neq_2)$ on the two degree 2 vertices connecting the two degree 4 vertices.
Then the constraint function $g_1$ of the construction $G_1$ is 
obtained by matrix multiplication $M(g_1) = M_{x_ix_j,x_sx_r}(g_1)
= M(f_1) \cdot N \cdot M(f_2)$,
where $N =  \left[\begin{smallmatrix} & & & 1\\ & & 1 & \\ & 1 & &\\ 1 & & & \end{smallmatrix}\right]$.
Thus
\[
M(g_1) = 
\left[\begin{smallmatrix} (b+c)d & 0 & 0 & bc+d^2 \\ 0 & a(b+c) & a^2+bc & 0 \\ 0 & a^2+bc & a(b+c) & 0 \\ bc+d^2 & 0 & 0 & (b+c)d \end{smallmatrix}\right].
\]
The  constraint function $g_1$ 
has four new parameters, denoted by 
\[(a_1, ~b_1, ~c_1, ~d_1) = (a(b+c), ~bc+d^2, ~a^2+bc, ~(b+c)d).\]
We make the following observations; all of them can be easily verified
using  (\ref{eqn:condition}): 
\begin{itemize}
\item
$d_1$ is the weight on sink and source and $0 < d_1 \le a_1, b_1, c_1$. 
\item
$c_1 = \max(a_1, b_1, c_1, d_1)$.
\item
$\frac{a_1}{d_1} = \frac{a}{d}, \frac{b_1}{d_1} \le \frac{b}{d}, \frac{c_1}{d_1} \le \frac{c}{d}$.
\item
$c_1d_1 \le a_1b_1$ because $c_1d_1 - a_1b_1 = -(b+c)(a-d)(bc - ad) \le 0$.
\end{itemize}
\item
$\boldsymbol{G_2}$: connect two vertices with constraint functions $f_2$ as in \figref{fig:2-chain}. Denote the constraint function of $G_2$ by $g_2$. We have
\[
M(g_2) = M(f_2) \cdot N \cdot M(f_2) =
\left[\begin{smallmatrix} 2cd & 0 & 0 & c^2+d^2 \\ 0 & 2ab & a^2+b^2 & 0 \\ 0 & a^2+b^2 & 2ab & 0 \\ c^2+d^2 & 0 & 0 & 2cd \end{smallmatrix}\right].
\]
The  constraint function $g_2$
has four new parameters, denoted by
\[(a_2, ~b_2, ~c_2, ~d_2) = (2ab, ~a^2 + b^2, ~c^2 + d^2, ~2cd).\]
The following observations can also be easily verified
using  (\ref{eqn:condition}):
\begin{itemize}
\item
$d_2$ is the weight on sink and source and if $cd \le ab$, then $0 < d_2 \le a_2, b_2, c_2$. 
\item
$\frac{a_2}{d_2} \le \frac{a}{d}, \frac{b_2}{d_2} \le \frac{b}{d}, \frac{c_2}{d_2} \le \frac{c}{d}$.
\item
$\frac{c_2 - d_2}{d_2} = \frac{(c-d)^2}{2cd} \le \frac{1}{2} \left(\frac{c-d}{d}\right)^2 \le \left(\frac{c-d}{d}\right)^2$.
\end{itemize}
\end{enumerate}

Based on the two basic constructions above, we construct
 the constraint function $\check{f}$ in logarithmically many rounds
recursively, each of the $O(\log n)$ rounds uses the constraint function
constructed in the previous round.
We now describe a single round in this construction, which consists
of two steps. In step 1 we use a signature with
some parameter setting $(a, b, c, d)$ satisfying (\ref{eqn:condition})
and apply  $G_1$ to two copies of the signature.
If the resulting parameter $b_1 < a_1$ we switch the roles of
$a_1$ and $b_1$, and obtain $(a'_1, b'_1, c'_1, d'_1)
= (b_1, a_1, c_1, d_1)$, again satisfying
(\ref{eqn:condition}), as well as $c_1d_1 \le a_1 b_1$.
In step 2, we apply $G_2$ to two copies of the signature constructed
in step 1 (with the switching of the roles of $a_1$ and $b_1$ if
it is needed).
Denote the parameters of the  resulting signature by $(a^*, b^*, c^*, d^*)$.
Altogether each round uses four copies of the signature from the
previous round, starting with the initial given signature.
Therefore in polynomial time we can afford to carry out $C \log n$ rounds
for any constant $C$.
Note that, 
if we consider the normalized quantities $(\frac{a}{d}, \frac{b}{d},
 \frac{c}{d}, \frac{d}{d})$,
then the respective quantities in each step $G_1$ and $G_2$ do
 not increase their distances to 1, i.e.,
 \[ 0 \le \frac{a^*}{d^*} - 1 \le \frac{a}{d} -1,~~~
0 \le \frac{b^*}{d^*} - 1 \le \frac{b}{d} -1,~~~
0 \le \frac{c^*}{d^*} - 1 \le \frac{c}{d} -1.\]
This is true even if the $G_2$ construction in step 2
is applied in the case when 
  the roles of $a_1$ and $b_1$ are switched for the signature from step 1,
when that switch is required ($b_1 < a_1$) as described.
More importantly,
based on the properties of $G_1$ and $G_2$, we know that
the (normalized) gap between $d$ and the previous largest entry $c$ shrinks quadratically fast, as measured by the new $c^*$ normalized with $d^*$.
More precisely,
\[ 0 \le \frac{c^*}{d^*} - 1 \le 
\left(\frac{c}{d} - 1 \right)^2.\]

Note that $c^*$ 
 may no longer be the largest among $a^*, b^*, c^*$;
however we will permute them to get $\widetilde{a}, \widetilde{b}, 
\widetilde{c}$ so that (\ref{eqn:condition}) is still satisfied before
proceeding to the next round.
This completes the description of our construction in one round
which obtains $(\widetilde{a}, \widetilde{b},
\widetilde{c}, \widetilde{d})$ from $(a,b,c,d)$.

We will construct the final signature $\check{f}$ 
by $O(\log n)$ rounds of this construction.
Also we will follow each value $a, b, c$ individually as they get transformed
through each round.
To state it formally, starting with the normalized triple
 $(\frac{a}{d}, \frac{b}{d},
 \frac{c}{d})$, we define a successor triple
$(\frac{a^*}{d^*}, \frac{b^*}{d^*},
 \frac{c^*}{d^*})$, so that each entry has the respective successor
(e.g., the entry $\frac{a}{d}$ has successor $\frac{a^*}{d^*}$).
This is well-defined because $(a_1, b_1, c_1, d_1)$ and $(a_2, b_2, c_2, d_2)$
are homogeneous functions of $(a,b,c,d)$.
Note that even though from one round to the next,
we may have to rename $a^*, b^*, c^*$ so that the permutated triple
$\widetilde{a}, \widetilde{b},
\widetilde{c}$ satisfies  (\ref{eqn:condition}),
the successor sequence as the rounds progress stays with an individual
value. E.g., starting from $(a,b,c,d)$,
if after one round $a^*  = \max(a^*, b^*, c^*, d^*) = \widetilde{c}$,
then the successor of $\frac{a}{d}$  after two rounds is 
$\frac{(\widetilde{c})^*}{(d^*)^*}$.
Now define $(\alpha_k, \beta_k, \gamma_k)$
to be the (ordered)  triple $(\frac{a}{d}, \frac{b}{d},
 \frac{c}{d})$,  or its successor triple,
at the beginning of the $k$-th round, for $k \ge 1$.


Let $\check{f}$ be the  4-ary signature constructed  after $3(k+1)$ rounds.
By the \emph{Pigeonhole Principle}, after $3(k+1)$ rounds, 
at least  one of $a,b, c$ has the property that
in at least $k+1$ many rounds  the corresponding  $\frac{a}{d}, \frac{b}{d},
 \frac{c}{d}$ or its successors are the maximum (normalized) value
in that round,
 and thus its next successor gets shrunken quadratically in that round.
Suppose this is $a$; the same proof works if it is $b$ or $c$. 
Let $\alpha_i$ be the maximum (normalized) value 
at the beginning of round $i$ in  $k+1$  rounds,
where $i \in \{i_0, \ldots, i_{k}\}$, and 
$1 \le i_0 < i_1 < \ldots < i_{k} \le 3(k+1)$.
Since initially we have $1 \le d \le a \le b \le c \le \frac{3}{2}d$,
\[0 \le \alpha_{i_1} - 1 \le \alpha_{i_0 +1} -1 \le 
\left( \alpha_{i_0} -1 \right)^2  
\le \frac{1}{2^2}.\]
Then
\[0 \le \alpha_{i_2} - 1 \le  \alpha_{i_1 +1} -1 
\le \left( \alpha_{i_1} -1 \right)^2 \le \frac{1}{2^{2^2}}.\]
By induction $0 \le \alpha_{i_k} - 1 \le \frac{1}{2^{2^k}}$.
At the end of  $3(k+1)$ rounds, 
if  $\check{f}$ has parameters
 $(\check{a},\check{b},\check{c},\check{d})$, then
\[0  \le  \frac{\max(\check{a},\check{b},\check{c})}{\check{d}}  -1
 \le  \alpha_{i_{k}} - 1 \le \frac{1}{2^{2^k}}.\]


Therefore, after logarithmically many rounds, using polynomially many vertices, we can get a 4-ary construction with parameters $\check{a}$, $\check{b}$, $\check{c}$, and $\check{d}$ that are exponentially close to 1
after normalizing by $\check{d}$.
Thus (\ref{eqn:result_2}) is proved.
\end{proof}

\section{\#\textsc{PerfectMatchings}-easiness}\label{sec:easiness}

In this section, we address two problems:
\begin{enumerate}
\item
What are the constraint functions that can be realized by 4-ary matchgates (\defnref{defn:matchgate})?

Although the set of constraint functions that can be realized by planar matchgates with complex edge weights have been completely characterized~\cite{cai_chen_2017},
the set of constraint functions that can be realized by general (not necessarily planar) matchgates with nonnegative real edge weights is not fully understood, even for matchgates of arity 4.
This type of matchgates plays a crucial role in the
study of the approximate complexity of counting problems, as we will
see in this paper.

In \thmref{thm:matchgate}, we give a complete characterization
of constraint functions of arity 4 that can be realized by 
matchgates  with nonnegative real edges.  Our method is primarily
geometric.
\item
\thmref{thm:pm-hard} shows that  
for positive parameters
$(a, b, c, d) \not\in \texttt{$d$-SUM}$
the problem \textup{\textsc{\#EightVertex($a, b, c, d$)}} is at least
as hard as  counting perfect matchings approximately. Here we ask the
reverse question:
For what parameter settings $(a, b, c, d)$ does \textup{\textsc{\#EightVertex($a, b, c, d$)}} $\le_\textup{AP}$ \textup{\textsc{\#PerfectMatchings}}?

We know that
\[Z_{\textup{\textsc{EightVertex}}}(G; a, b, c, d) = \textup{Holant}\left(G'; \neq_2 |\ f\right),\]
where $f$ is the 4-ary constraint function with $M(f) = \left[\begin{smallmatrix} d & 0 & 0 & a \\ 0 & b & c & 0 \\ 0 & c & b & 0 \\ a & 0 & 0 & d \end{smallmatrix}\right]$.
Considering the fact that $(\neq_2)$ can be easily realized by a matchgate (a vertex with two dangling edges),
\thmref{thm:pm-easy} is a direct consequence of \lemref{lem:pm-easy} which says that any constraint function in \texttt{SQ-SUM} 
is realizable by some \emph{4-ary matchgate} of constant size
 (with nonnegative edge weights, but not necessarily planar) 
(see \defnref{defn:matchgate}).
Our theorem works for the eight-vertex model with parameter settings $\mathcal{S}^\textup{E}_{\le^2}$ (defined below) not necessarily satisfying the arrow reversal symmetry.
 
Moreover, \lemref{lem:pm-easy_tight} indicates that our result is tight in the sense that $\mathcal{S}^\textup{E}_{\le^2}$ captures precisely
the set of all constraint functions that can be realized by 4-ary matchgates 
(with even support, i.e., nonzero only on inputs of even Hamming weight).
A similar statement holds for $\mathcal{S}^\textup{O}_{\le^2}$.
the corresponding set with odd support.

\end{enumerate}

\begin{definition}\label{defn:matchgate}
We use the term \emph{a $k$-ary matchgate} to denote a graph $\Gamma$ having $k$ ``dangling'' edges, labelled $i_1, \dots, i_k$.
Each dangling edge has weight $1$ and each non-dangling edge $e$ is equipped with a 
nonnegative weight $w_e$.
A configuration is a $0, 1$-assignment to the edges. A configuration is a perfect matching if every vertex has exactly one incident edge assigned $1$. The matchgate implements the constraint function $f$, where $f(b_1, \dots, b_k)$ for
$(b_1, \dots, b_k) \in \{0, 1\}^k$ is the sum, over perfect matchings, of the product of the weight of edges with assignment $1$, where 
the dangling edge $i_j$ is assigned $b_j$, and the empty product has weight $1$.
\end{definition}

\begin{remark}
Contrary to \defnref{defn:matchgate} which does not require planarity, planar matchgates with complex edge weights has been completely characterized~\cite{doi:10.1137/S0097539700377025, cai_chen_2017}.
As computing the weighted sum of perfect matchings is in polynomial time over planar graphs by the \emph{FKT algorithm}~\cite{doi:10.1080/14786436108243366, KASTELEYN19611209, kasteleyn_book}, problems that can be locally expressed by planar matchgates are tractable over planar graphs.
\end{remark}

\begin{notation}
$\mathcal{S}^\textup{E}_{\le^2} = \{f \ |\ M(f) = \left[\begin{smallmatrix} d_1 & 0 & 0 & a_1 \\ 0 & b_1 & c_1 & 0 \\ 0 & c_2 & b_2 & 0 \\ a_2 & 0 & 0 & d_2 \end{smallmatrix}\right] \text{ satisfying }
\left\{\begin{smallmatrix}
a_1a_2 \ \le\ b_1b_2 + c_1c_2 + d_1d_2 \\
b_1b_2 \ \le\ a_1a_2 + c_1c_2 + d_1d_2 \\
c_1c_2 \ \le\ a_1a_2 + b_1b_2 + d_1d_2 \\
d_1d_2 \ \le\ a_1a_2 + b_1b_2 + c_1c_2
\end{smallmatrix}\right.,
a_1, a_2, b_1, b_2, c_1, c_2, d_1, d_2 \ge 0.
\}$,
$\mathcal{S}^\textup{O}_{\le^2} = \{f \ |\ M(f) = \left[\begin{smallmatrix} 0 & d_1 & a_1 & 0 \\ b_1 & 0 & 0 & c_1 \\ c_2 & 0 & 0 & b_2 \\ 0 & a_2 & d_2 & 0 \end{smallmatrix}\right] \text{ satisfying }
\left\{\begin{smallmatrix}
a_1a_2 \ \le\ b_1b_2 + c_1c_2 + d_1d_2 \\
b_1b_2 \ \le\ a_1a_2 + c_1c_2 + d_1d_2 \\
c_1c_2 \ \le\ a_1a_2 + b_1b_2 + d_1d_2 \\
d_1d_2 \ \le\ a_1a_2 + b_1b_2 + c_1c_2
\end{smallmatrix}\right.,
a_1, a_2, b_1, b_2, c_1, c_2, d_1, d_2 \ge 0.
\}$.
\end{notation}

\begin{theorem}\label{thm:matchgate}
Denote by $\mathcal{M}$ the set of constraint functions that can be realized by 4-ary matchgates.
Then $\mathcal{M} = \mathcal{S}^\textup{E}_{\le^2} \bigcup \mathcal{S}^\textup{O}_{\le^2}$.
\end{theorem}
\begin{remark}
Note that any constraint function in $\mathcal{M}$ must satisfy either \emph{even parity} (nonzero only on inputs of even Hamming weight) or \emph{odd parity} (nonzero only on inputs of odd Hamming weight).
\thmref{thm:matchgate} for the even parity part ($\mathcal{S}^\textup{E}_{\le^2}$) is a combination of \lemref{lem:pm-easy} and \lemref{lem:pm-easy_tight}.
The odd parity part can be proved similarly.
\end{remark}


\begin{lemma}\label{lem:pm-easy}
Suppose $f \in \mathcal{S}^\textup{E}_{\le^2}$.
Then there is a 4-ary matchgate of constant size whose constraint function is $f$.
\end{lemma}
\begin{proof}
\captionsetup[subfigure]{labelformat=parens}
\begin{figure}[h!]
\centering
\begin{subfigure}[b]{0.3\linewidth}
\centering\includegraphics[width=\linewidth]{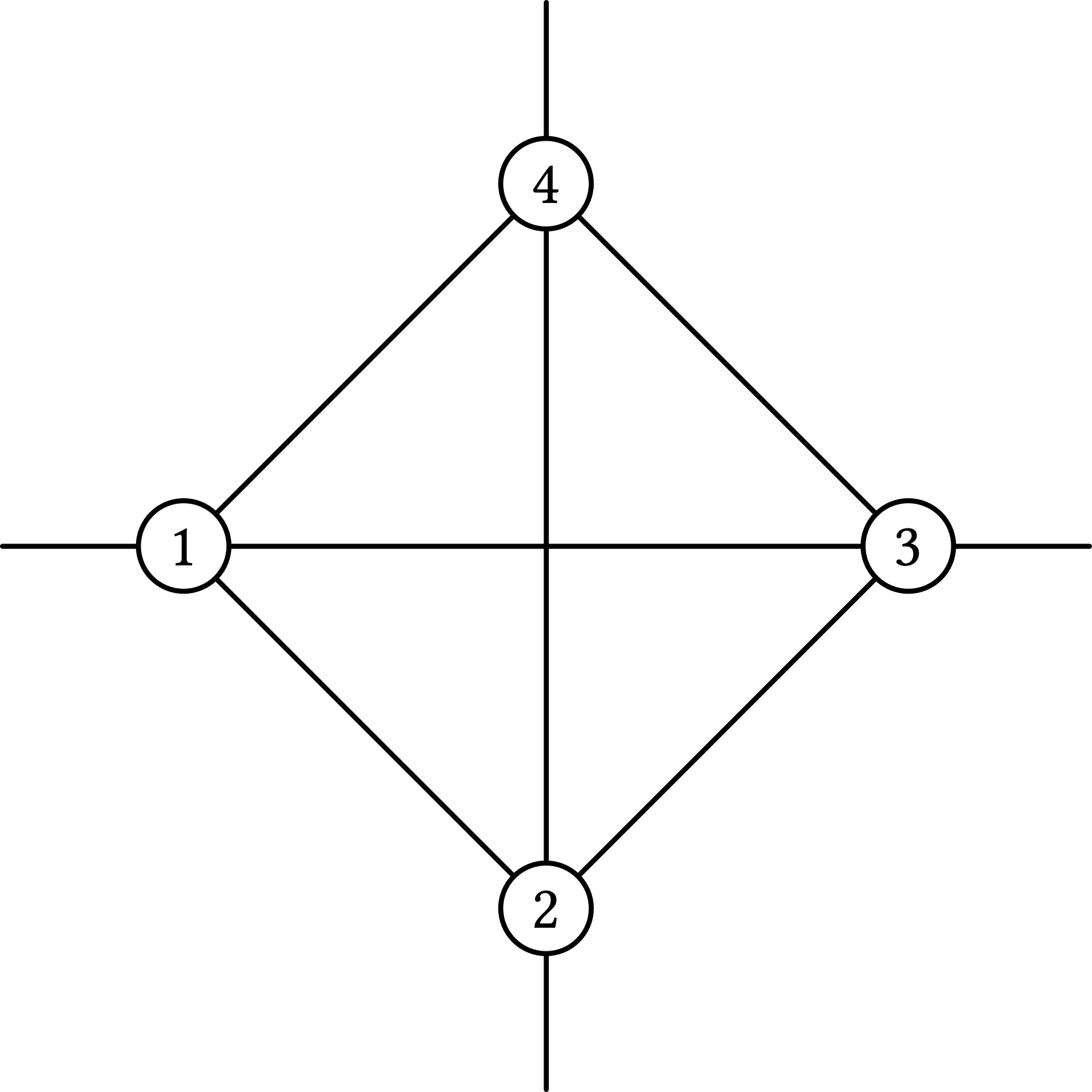}\caption{}\label{fig:matchgate_d}
\end{subfigure}
\hspace{0.03\linewidth}
\begin{subfigure}[b]{0.3\linewidth}
\centering\includegraphics[width=\linewidth]{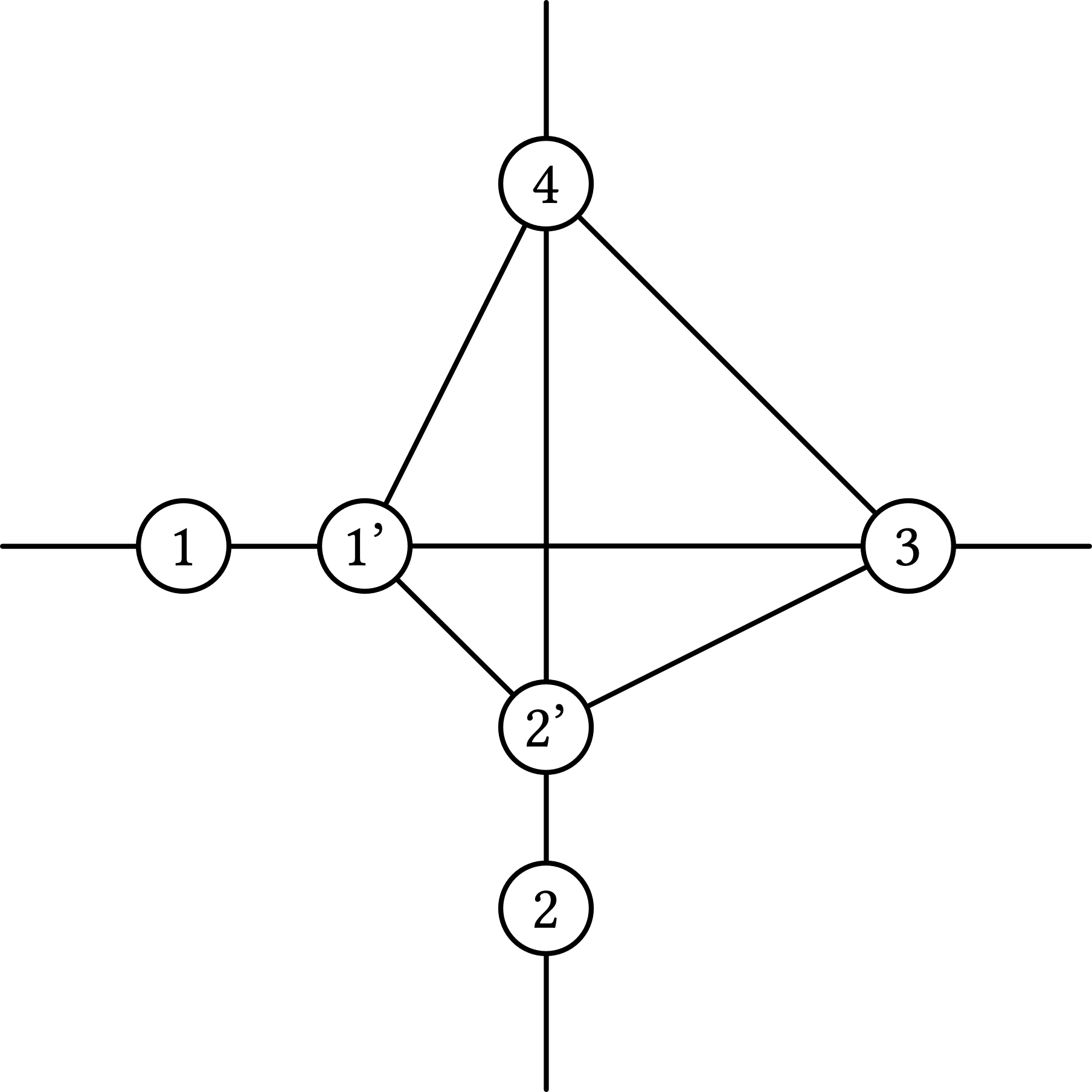}\caption{}\label{fig:matchgate_a}
\end{subfigure}
\hspace{0.03\linewidth}
\begin{subfigure}[b]{0.3\linewidth}
\centering\includegraphics[width=\linewidth]{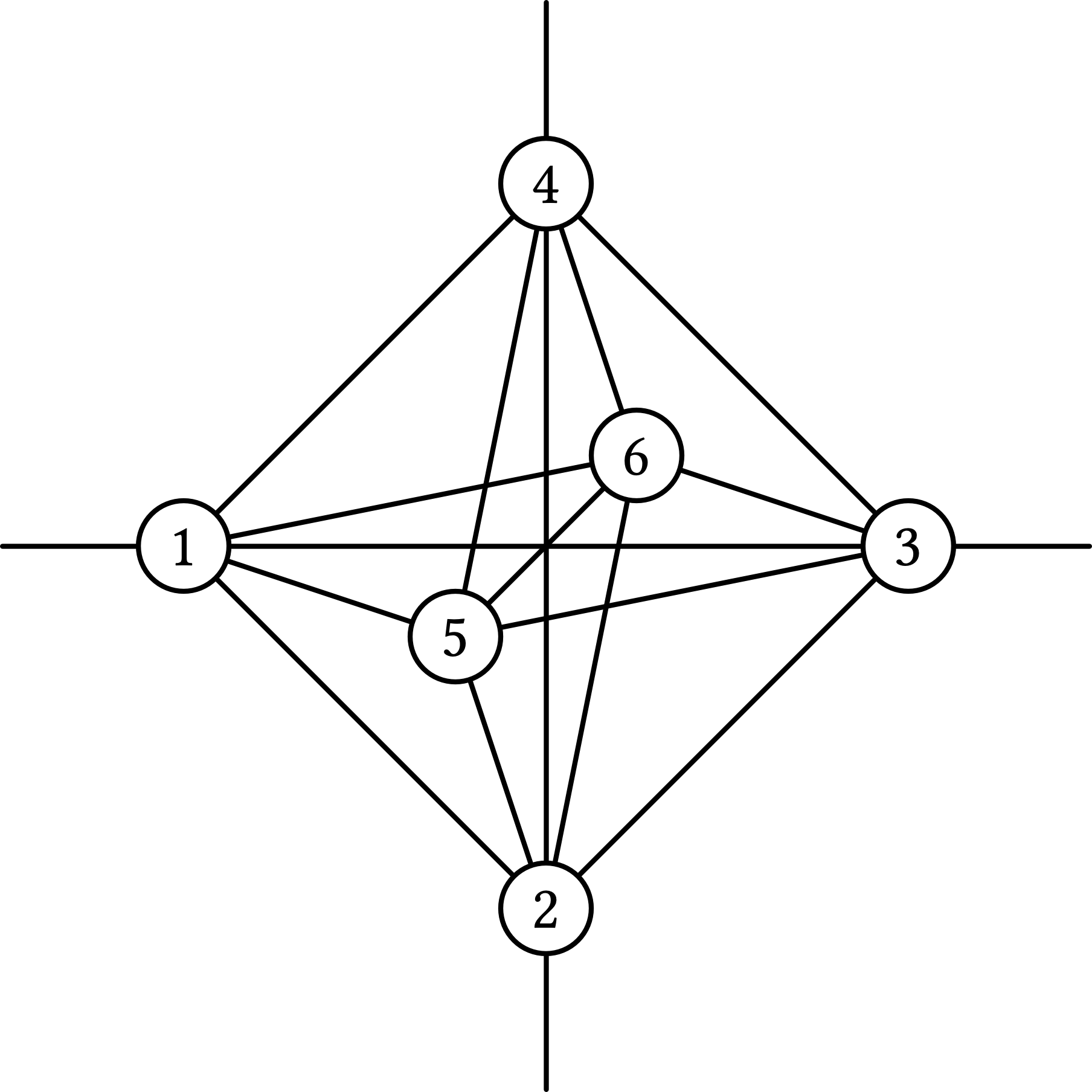}\caption{}\label{fig:matchgate_interior}
\end{subfigure}
\caption{4-ary matchgates.}\label{fig:matchgates}
\end{figure}

We first note that if any of the four inequalities in the definition of
$\mathcal{S}_{\le^2}$ is an equality, then  the remaining three
inequalities automatically hold, since the 8 values $a_1, \ldots, d_2$
 are all nonnegative.

Given a constraint function $\left[\begin{smallmatrix} d_1 & 0 & 0 & a_1 \\ 0 & b_1 & c_1 & 0 \\ 0 & c_2 & b_2 & 0 \\ a_2 & 0 & 0 & d_2 \end{smallmatrix}\right]$,
first we construct a matchgate for $d_1d_2 = a_1a_2 + b_1b_2 + c_1c_2$.
If $d_1d_2=0$ then all four products $a_1a_2 = b_2b_2=c_1c_2= d_1d_2 =0$,
and one can easily adapt from the following proof to show that the signature is
realizable as a matchgate signature. 
So it suffices to implement the normalized version
$\left[\begin{smallmatrix} a_1a_2 + b_1b_2 + c_1c_2 & 0 & 0 & a_1 \\ 0 & b_1 & c_1 & 0 \\ 0 & c_2 & b_2 & 0 \\ a_2 & 0 & 0 & 1 \end{smallmatrix}\right]$.
Our construction is a weighted $K_4$ depicted in \figref{fig:matchgate_d}.
Let $e_1, e_2, e_3, e_4$ be the dangling edges incident to vertices $1, 2, 3, 4$, respectively.
Denote by $w_{ij}$ the weight on the edge between vertex $i$ and vertex $j$.
One can check that the following weight assignment meets our need: $w_{12} = a_1, w_{34} = a_2, w_{14} = b_1, w_{23} = b_2, w_{13} = c_1, w_{24} = c_2$.

For $a_1a_2 = b_1b_2 + c_1c_2 + d_1d_2$,
without loss of generality we assume $a_1a_2 \not =0$
and we normalize $a_1=1$.
Then our construction is shown in \figref{fig:matchgate_a} where we set
$w_{11'} = 1, w_{22'} = 1, w_{1'2'} = d_2, w_{34} = d_1, w_{1'4} = c_2, w_{2'3} = c_1, w_{1'3} = b_2, w_{2'4} = b_1$.
One can verify that it realizes the normalized constraint function
$\left[\begin{smallmatrix} d_1 & 0 & 0 & 1 \\ 0 & b_1 & c_1 & 0 \\ 0 & c_2 & b_2 & 0 \\ b_1b_2 + c_1c_2 + d_1d_2 & 0 & 0 & d_2 \end{smallmatrix}\right]$.
The construction for $b_1b_2 = a_1a_2 + c_1c_2 + d_1d_2$ and $c_1c_2 = a_1a_2 + b_1b_2 + d_1d_2$ are symmetric to the above case.

It remains to show that the interior
\begin{equation}\label{eqn:interior_1}
\left\{\begin{matrix}
a_1a_2 & < & b_1b_2 + c_1c_2 + d_1d_2 \\
b_1b_2 & < & a_1a_2 + c_1c_2 + d_1d_2 \\
c_1c_2 & < & a_1a_2 + b_1b_2 + d_1d_2 \\
d_1d_2 & < & a_1a_2 + b_1b_2 + c_1c_2
\end{matrix}\right.
\end{equation}
can all be reached.
We first deal with the case when all eight parameters are strictly positive and leave the other cases to the end of this proof.
We use  a weighted $K_6$ to be our matchgate depicted in \figref{fig:matchgate_interior},
and set $w_{12} = r_1, w_{34} = r_2, w_{14} = s_1, w_{23} = s_2, w_{13} = t_1, w_{24} = t_2, w_{15} = p_1, w_{25} = p_2, w_{35} = p_3, w_{45} = p_4, w_{16} = q_1, w_{26} = q_2, w_{36} = q_3, w_{46} = q_4, w_{56} = 1$.
Then the matchgate has a singature with the following parameters
\begin{align*}
a'_1 & = r_1 + p_1q_2 + p_2q_1, &
a'_2 & = r_2 + p_3q_4 + p_4q_3,\\
b'_1 & = s_1 + p_1q_4 + p_4q_1, &
b'_2 & = s_2 + p_2q_3 + p_3q_2,\\
c'_1 & = t_1 + p_1q_3 + p_3q_1, &
c'_2 & = t_2 + p_2q_4 + p_4q_2,\\
d'_1 & = (r_1r_2 + s_1s_2 + t_1t_2) + & d'_2 & =  1,\\
 &~~~~(p_3q_4 + p_4q_3)r_1 + (p_1q_2 + p_2q_1) r_2 + & \\
 &~~~~(p_2q_3 + p_3q_2) s_1 + (p_1q_4 + p_4q_1) s_2 + & \\
 &~~~~(p_2q_4 + p_4q_2) t_1 + (p_1q_3 + p_3q_1) t_2. &
\end{align*}
Note that all the edge weights have to be nonnegative.
By properly setting the edge weights in the matchgate, we show that we can achieve any relative ratios among the eight given positive
 values $a_1, a_2, b_1, b_2, c_1, c_2, d_1, d_2$ that satisfy (\ref{eqn:interior_1}).
Our first step is to achieve any relative ratios among the four
product values $a_1a_2$, $b_1b_2$, $c_1c_2$, $d_1d_2$ satisfying (\ref{eqn:interior_1});
and the second step is to adjust the relative ratio within the pairs
 $\{a_1, a_2\}$, $\{b_1, b_2\}$,  $\{d_1, d_2\}$  and $\{c_1, c_2\}$
without affecting the product values.
This can be justified by the observation that,
by a scaling a global positive constant can be easily achieved, and
all appearances of $a_1$ and $a_2$ in (\ref{eqn:interior_1}) are as a product $a_1a_2$, and similarly for $b_1, b_2, c_1, c_2, d_1, d_2$.

For the fourteen edge weights $r_1, \ldots, q_4$ to be determined,
let
\begin{equation}\label{eqn:ratio_2}
\left\{\begin{matrix}
A'  = p_1p_2q_3q_4 + p_3p_4q_1q_2,
& R = r_1r_2 + r_1(p_3q_4 + p_4q_3) + r_2(p_1q_2 + p_2q_1), \\
B'  = p_1p_4q_2q_3 + p_2p_3q_1q_4,
& S = s_1s_2 + s_1(p_2q_3 + p_3q_2) + s_2(p_1q_4 + p_4q_1), \\
C'  = p_1p_3q_2q_4 + p_2p_4q_1q_3,
& T = t_1t_2 + t_1(p_2q_4 + p_4q_2) + t_2(p_1q_3 + p_3q_1),
\end{matrix}\right.
\end{equation}
and 
define
\begin{equation}\label{eqn:ratio_1}
\left\{\begin{matrix}
A &=&  A' + S + T\\
B &=&  B' + R + T\\
C &=&  C' + R + S\\
D &=&  ~~~A' + B' + C'.
\end{matrix}\right.
\end{equation}
Note that $A', B', C', R, S, T$ are all nonnegative and so are $A, B, C, D$.

Our goal is to choose the fourteen edge weights $r_1, \ldots, q_4$ 
so that $A, B, C, D$ are all positive and satisfy
\begin{equation}\label{eqn:ratio_3}
\left\{\begin{matrix}
A &= \frac{1}{2}(b_1b_2 + c_1c_2 + d_1d_2 - a_1a_2) \\
B &= \frac{1}{2}(a_1a_2  + c_1c_2 + d_1d_2 - b_1b_2) \\
C &= \frac{1}{2}(a_1a_2 + b_1b_2 + d_1d_2 - c_1c_2) \\
D &= \frac{1}{2}(a_1a_2 + b_1b_2 + c_1c_2 - d_1d_2). 
\end{matrix}\right.
\end{equation}
Note that, by definition, the left-side of (\ref{eqn:ratio_3}) is precisely
the right-side of (\ref{eqn:ratio_3}) when $a_1, \ldots, d_2$ are
replaced by $a'_1, \ldots, d'_2$ respectively.
Denote the products $a_1a_2, b_1b_2, c_1c_2, d_1d_2$ by $a^{**}, b^{**}, c^{**}, d^{**}$ respectively.
Then  (\ref{eqn:ratio_3}) is a set of four linear equations
$M \cdot \left[\begin{smallmatrix} a^{**} \\ b^{**} \\ c^{**} \\ d^{**} \end{smallmatrix}\right] = \left[\begin{smallmatrix} A \\ B \\ C \\ D \end{smallmatrix}\right]$,
where $M = \frac{1}{2}\left[\begin{smallmatrix} -1 & 1 & 1 & 1 \\ 1 & -1 & 1 & 1 \\ 1 & 1 & -1 & 1 \\ 1 & 1 & 1 & -1 \end{smallmatrix}\right]$.
Note that $M$ is invertible and $M^{-1} = M$, so 
(\ref{eqn:ratio_3}) is equivalent to
$M \cdot \left[\begin{smallmatrix} A \\ B \\ C \\ D \end{smallmatrix}\right] = \left[\begin{smallmatrix} a^{**} \\ b^{**} \\ c^{**} \\ d^{**} \end{smallmatrix}\right]$,
having an identical form.
Since the requirement (\ref{eqn:interior_1}) in terms of $a^{**}, b^{**}, c^{**}, d^{**}$ translates into the requirement $A, B, C, D$ being strictly positive via $M$, it is not surprising that the requirement $a^{**}, b^{**}, c^{**}, d^{**}$ being strictly positive translates into the requirement
\begin{equation}\label{eqn:interior_2}
\left\{\begin{matrix}
A < B + C + D\\
B < A + C + D\\
C < A + B + D\\
D < A + B + C,
\end{matrix}\right.
\end{equation}
and that
$A, B, C, D$ are positive is the same as (\ref{eqn:interior_1}).

Furthermore, let
$\left\{\begin{smallmatrix}
X = S + T \\
Y = R + T \\
Z = R + S
\end{smallmatrix}\right.$, then
 the requirement $R, S, T$ being positive is equivalent to the requirement
$\left\{\begin{smallmatrix}
Y + Z > X \\
X + Z > Y \\
X + Y > Z
\end{smallmatrix}\right.$.
This is because
$\left[\begin{smallmatrix} 0 & 1 & 1 \\ 1 & 0 & 1 \\ 1 & 1 & 0\end{smallmatrix}\right] \cdot \left[\begin{smallmatrix} R \\ S \\ T \end{smallmatrix}\right] = \left[\begin{smallmatrix} X \\ Y \\ Z \end{smallmatrix}\right]$
is the same as
$\frac{1}{2}\left[\begin{smallmatrix} -1 & 1 & 1 \\ 1 & -1 & 1 \\ 1 & 1 & -1\end{smallmatrix}\right] \cdot \left[\begin{smallmatrix} X \\ Y \\ Z \end{smallmatrix}\right] = \left[\begin{smallmatrix} R \\ S \\ T \end{smallmatrix}\right]$.

The crucial ingredient of our proof is a 
\emph{geometric} lemma in  3-dimensional space.
Suppose  $a^{**}, b^{**}, c^{**}, d^{**}$ are positive and they satisfy
 (\ref{eqn:interior_1}).
This defines 
$\left[\begin{smallmatrix} \tilde{A}\\ \tilde{B}\\ \tilde{C}\\ \tilde{D}
 \end{smallmatrix}\right] =
M \cdot \left[\begin{smallmatrix} a^{**} \\ b^{**} \\ c^{**} \\ d^{**} \end{smallmatrix}\right]$.
By a scaling we may assume $\tilde{D} =1$.
Hence $(\tilde{A}, \tilde{B}, \tilde{C}, \tilde{D})$ are positive and
satisfy (\ref{eqn:interior_2}).
Thus $(\tilde{A}, \tilde{B}, \tilde{C})$ belongs to the set
$U$ in the statement of \lemref{lem:pm-easy_geometry}.

By \lemref{lem:pm-easy_geometry},
there exist (strictly) positive tuples $(\tilde{A}', \tilde{B}', \tilde{C}')$ and $(\tilde{X}, \tilde{Y}, \tilde{Z})$ 
such that
\[(\tilde{A}, \tilde{B}, \tilde{C}) = (\tilde{A}', \tilde{B}', \tilde{C}') + (\tilde{X}, \tilde{Y}, \tilde{Z}),\]
satisfying
$\tilde{A}' + \tilde{B}' + \tilde{C}' = 1$ and
$\left\{\begin{smallmatrix}
\tilde{Y} + \tilde{Z} > \tilde{X}\\
\tilde{X} + \tilde{Z} > \tilde{Y}\\
\tilde{X} + \tilde{Y} > \tilde{Z}
\end{smallmatrix}\right.$.
By the previous observation this indicates that
there exist (strictly) positive $\tilde{A}', \tilde{B}', \tilde{C}', \tilde{R}, \tilde{S}, \tilde{T}$ such that
$\left\{\begin{smallmatrix}
\tilde{A} =  \tilde{A}' + \tilde{S} + \tilde{T}\\
\tilde{B} =  \tilde{B}' + \tilde{R} + \tilde{T}\\
\tilde{C} =  \tilde{C}' + \tilde{R} + \tilde{S}\\
\tilde{D} =  \tilde{A}' + \tilde{B}' + \tilde{C}'
\end{smallmatrix}\right.$.

We first set $p_i, q_i \ (1 \le i \le 4)$ so that $(A', B', C') = c \cdot (\tilde{A}', \tilde{B}', \tilde{C}')$ for some constant $c$.
To achieve this, set $q_1 = q_2 = q_3 = q_4 = 1$,
and let  $o_1, o_2, o_3$ be positive, and then
set $p_1 = \sqrt{\frac{o_2 o_3}{o_1}}$,
$p_2 = \sqrt{\frac{o_3 o_1}{o_2}}$,
$p_3 = \sqrt{\frac{o_1 o_2}{o_3}}$,
and $p_4 = \frac{1}{p_1p_2p_3}$.
We have  $p_1p_2p_3p_4 = 1$, and  $p_2 p_3 = o_1,
p_3p_1 = o_2, p_1p_2 = o_3$.
  Then set
$\left\{\begin{smallmatrix}
A'  = p_1p_2 + \frac{1}{p_1p_2} = o_3 + \frac{1}{o_3} \\
B'  = p_2p_3 + \frac{1}{p_2p_3} = o_1 + \frac{1}{o_1} \\
C'  = p_3p_1 + \frac{1}{p_3p_1} = o_2 + \frac{1}{o_2}
\end{smallmatrix}\right.$, which can be independently any positive numbers  
at least 2,
by choosing
$o_1, o_2, o_3$ to be suitable positive numbers.
This allows us to get $A', B', C'$ such that $(A', B', C') = c \cdot (\tilde{A}', \tilde{B}', \tilde{C}')$ for some constant $c$.
Then it is obvious that $r_1, r_2, s_1, s_2, t_1, t_2$ can be set so that $(R, S, T) = c \cdot (\tilde{R}, \tilde{S}, \tilde{T})$.
Compute $A, B, C, D$ according to (\ref{eqn:ratio_1}).
As a consequence, $(A, B, C, D) = c \cdot (\tilde{A}, \tilde{B}, \tilde{C}, \tilde{D})$ is a valid solution.

To adjust the relative ratio between $\{d_1, d_2\}$, say increasing $\frac{d_2}{d_1}$ by $\delta$, while keeping all product values and
the relative ratios within the other three pairs, just increase $r_1, r_2, s_1, s_2, t_1, t_2$ by $\delta^{1/2}$ and increase $p_i, q_i \ (1 \le i \le 4)$ by $\delta^{1/4}$.
Similarly, to increase $\frac{a_2}{a_1}$ by $\delta$ alone
without affecting the other products and ratios, just increase $r_2$ by $\delta^{1/2}$ and $p_3, p_4, q_3, q_4$ by $\delta^{1/4}$, and decrease $r_1$ by $\delta^{1/2}$ and $p_1, p_2, q_1, q_2$ by $\delta^{1/4}$.
The other cases are symmetric.

Finally we deal with the cases when there are zeros among the eight parameters.
Note that at most one of $a_1a_2, b_1b_2, c_1c_2, d_1d_2$ is zero, 
because if at least two products are zero, say $a_1a_2 = b_1b_2 = 0$, then 
(\ref{eqn:interior_1}) forces a contradiction 
that  $c_1c_2 < d_1d_2$ and $d_1d_2 < c_1c_2$.
In the case $d_1d_2 = 0$:
\begin{itemize}
\item
$d_1 = 0, d_2 \neq 0$: We make the modification that $w_{12} = w_{34} = w_{14} = w_{23} = w_{13} = w_{24} = 0$, i.e. $r_1, r_2, s_1, s_2, t_1, t_2 = 0$.
\item
$d_1 = d_2 = 0$: We make the further modification that $w_{56} = 0$.
\item
$d_1 \neq 0, d_2 = 0$: We connect the four dangling edges in \figref{fig:matchgate_interior} to four degree 2 vertices, respectively. This switches the role of $d_1$ and $d_2$ in the previous proof.
\end{itemize}
One can check our proof is still valid in the above three cases.
If $a_1a_2 = 0$, then we connect the dangling edges on vertices 1, 2 to two degree 2 vertices (similar to the operation from \figref{fig:matchgate_d} to \figref{fig:matchgate_a}). This switches the role of $d_1, d_2$ with $a_2, a_1$ and the proof folllows.
The proofs for $b_1b_2 = 0$ and $c_1c_2 = 0$ are symmetric.
\end{proof}

Now we give the crucial geometric lemma.
\begin{lemma}\label{lem:pm-easy_geometry}
Let $U = \{(x, y, z) \in \mathbb{R}^3_{>0} \ |\ 
x < y + z + 1,
y < x + z + 1,
z < x + y + 1,
1 < x + y + z
\}$,
$V = \{(x, y, z) \in \mathbb{R}^3_{>0} \ |\ x + y + z = 1\}$, and
$W = \{(x, y, z) \in \mathbb{R}^3_{>0} \ |\ 
y + z > x, 
x + z > y, 
x + y > z
\}$.
Then $U$ is the Minkowski sum of $V$ and $W$,
namely, $U$ consists of precisely those points 
$\boldsymbol{u} \in {\mathbb{R}}^3$,
such that $\boldsymbol{u} = \boldsymbol{v} + \boldsymbol{w}$
for some $\boldsymbol{v} \in V$ and $\boldsymbol{w} \in W$.
The same statement is true for the closures of $U$, $V$ and $W$.
\end{lemma}
\begin{proof}
Observe that $U$, $V$, and $W$ are the interiors of a polyhedron with
7 facets, a regular triangle, and a polyhedron with 3 facets,
respectively. 


The polyhedron for $W$ is the intersection of three
half spaces    bounded by three planes,
$\left\{\begin{smallmatrix}
(\pi_1): y + z \ge x \\
(\pi_2): x + z \ge y \\
(\pi_3): x + y \ge z
\end{smallmatrix}\right.$,
where the planes are defined by equalities.
Note that these inequalities imply that $x, y, z \ge 0$, thus
this polyhedron has only three facets.
We can find the intersection of each pair of the three planes for $W$ as
$\left\{\begin{smallmatrix}
\pi_1 \cap \pi_2: x = y \ge 0, z = 0 \\
\pi_1 \cap \pi_3: x = z \ge 0, y = 0 \\
\pi_2 \cap \pi_3: y = z \ge 0, x = 0
\end{smallmatrix}\right.$.
Note that these intersections lie on the planes $z = 0$, $y = 0$, and $x = 0$, respectively.

\captionsetup[subfigure]{labelformat=parens}
\begin{figure}[h!]
\centering
\begin{subfigure}[t]{0.42\linewidth}
\includegraphics[width=0.8\linewidth]{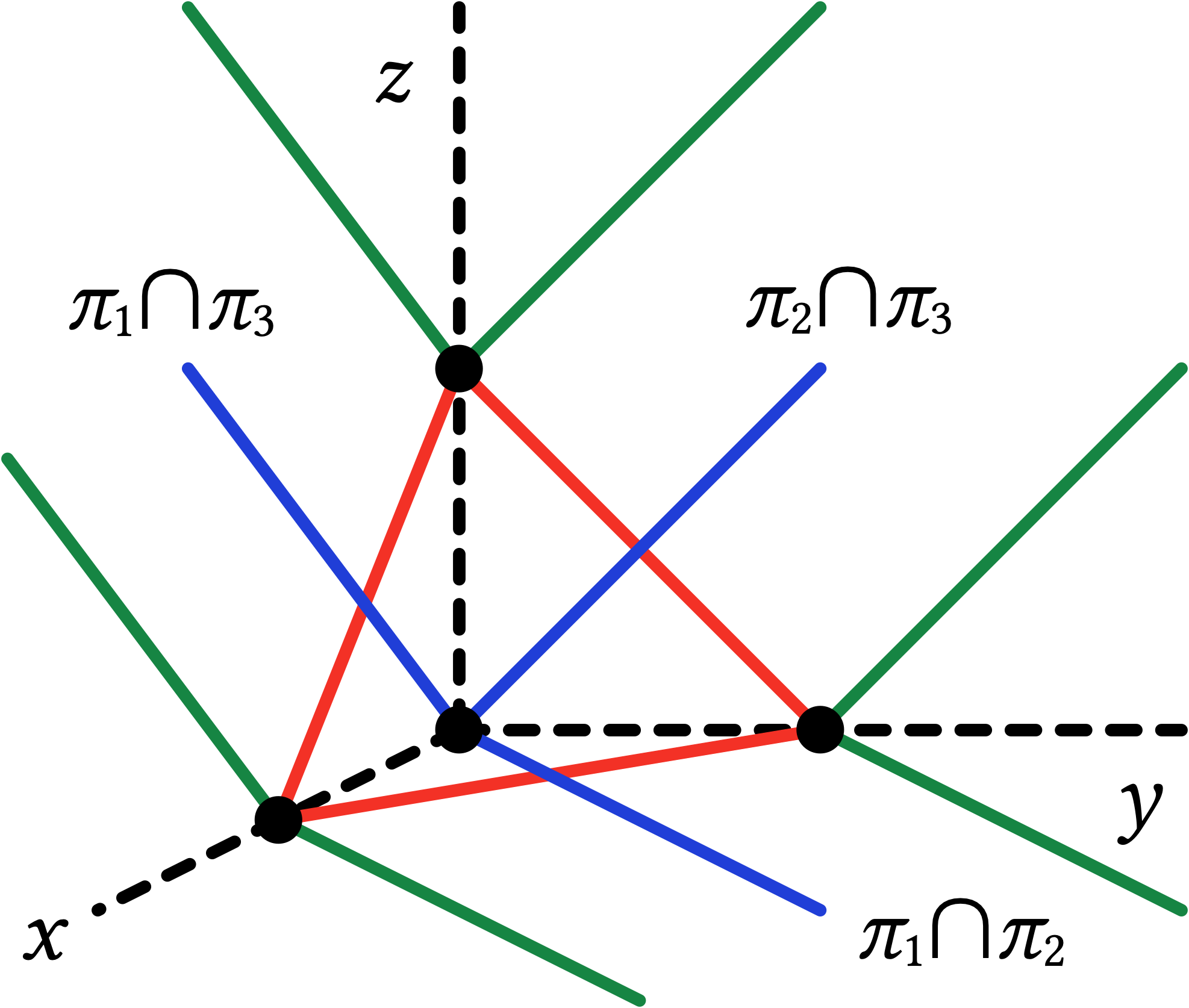}
\caption{The blue rays are the intersections of the three facets for $W$. The red triangle is the boundary for $V$. The green rays together with the red triangle are the intersections of the seven facets for $U$.}
\label{fig:geometry}
\end{subfigure}
\hspace{0.05\linewidth}
\begin{subfigure}[t]{0.42\linewidth}
\centering
\includegraphics[width=0.8\linewidth]{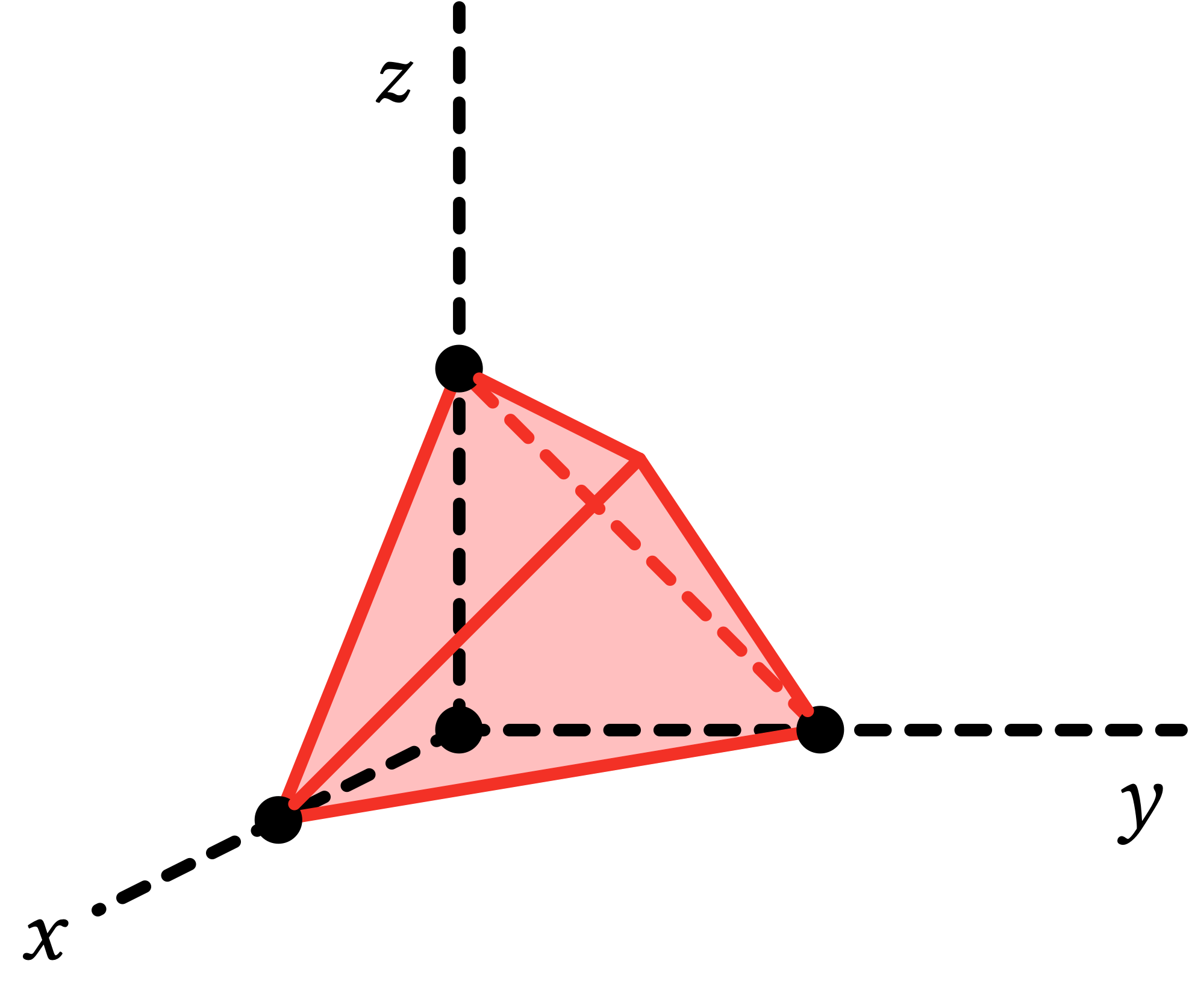}
\caption{The tetrahedron in $U$ that is left uncovered by
sliding $W$ along the boundary of $V$, but is covered by the rays
from the simplex $V$ in the direction of $(1,1,1)$.}
\label{fig:tetrahedron}
\end{subfigure}
\caption{}
\end{figure}

Let $\boldsymbol{e}_1 = (1, 0, 0), \boldsymbol{e}_2 = (0, 1, 0), \boldsymbol{e}_3 = (0, 0, 1)$.
Then the triangle for $V$ is just the convex hull of $\boldsymbol{e}_1, \boldsymbol{e}_2, \boldsymbol{e}_3$.
Suppose we shift the origin of  $W$ from $\boldsymbol{0}$ to $\boldsymbol{e}_1$ and denote the resulting (interior of a) polyhedron by $W^{\boldsymbol{e}_1}$,
then we have the defining inequalities
$\left\{\begin{smallmatrix}
(\pi_1^{\boldsymbol{e}_1}): y + z \ge (x - 1) \\
(\pi_2^{\boldsymbol{e}_1}): (x - 1) + z \ge y \\
(\pi_3^{\boldsymbol{e}_1}): (x - 1) + y \ge z
\end{smallmatrix}\right.$,
where the shifted planes are defined by the corresponding equalities.
By symmetry, if we shift  the origin of $W$ to 
 $\boldsymbol{e}_2$ and to $\boldsymbol{e}_3$, we have
respectively  $W^{\boldsymbol{e}_2}$ with
$\left\{\begin{smallmatrix}
(\pi_1^{\boldsymbol{e}_2}): (y - 1) + z \ge x \\
(\pi_2^{\boldsymbol{e}_2}): x + z \ge (y - 1) \\
(\pi_3^{\boldsymbol{e}_2}): x + (y - 1) \ge z
\end{smallmatrix}\right.$, and
$W^{\boldsymbol{e}_3}$ with
$\left\{\begin{smallmatrix}
(\pi_1^{\boldsymbol{e}_3}): y + (z - 1) \ge x \\
(\pi_2^{\boldsymbol{e}_3}): x + (z - 1) \ge y \\
(\pi_3^{\boldsymbol{e}_3}): x + y \ge (z - 1)
\end{smallmatrix}\right.$.
Note that the shifted planes
 $\pi_1^{\boldsymbol{e}_1}$, $\pi_2^{\boldsymbol{e}_2}$, and $\pi_3^{\boldsymbol{e}_3}$ contain three distinct
facets of $U$, and they coincide
exactly with a facet of  $W^{\boldsymbol{e}_1}$, $W^{\boldsymbol{e}_2}$, and $W^{\boldsymbol{e}_3}$, respectively.

By sliding $W$ with its origin along the line   
 $x + y = 1, z = 0$  from $\boldsymbol{e}_1$ to $\boldsymbol{e}_2$,
  we have
a partial coverage of $U$ by the shifting copies of
 $W$ from $W^{\boldsymbol{e}_1}$
and $W^{\boldsymbol{e}_2}$:
\begin{itemize}
\item
The shifted ray of  $\pi_1 \cap \pi_2: x = y \ge 0, z = 0$
moves from 
$\pi_1^{\boldsymbol{e}_1} \cap \pi_2^{\boldsymbol{e}_1}:
 (x - 1) = y \ge 0, z = 0$
to $\pi_1^{\boldsymbol{e}_2} \cap \pi_2^{\boldsymbol{e}_2}:
 x = (y - 1) \ge 0, z = 0$. Notice that this is a parallel transport,
and stays on the plane $z=0$, and thus it 
swipes another facet of $U$ on $z = 0$ bounded by
the two lines $x-y =-1$ and $x-y =1$.
\item
The shifted ray of 
$\pi_1 \cap \pi_3: x = z \ge 0, y=0$ moves from 
$\pi_1^{\boldsymbol{e}_1} \cap \pi_3^{\boldsymbol{e}_1}: (x - 1) = z \ge 0, y = 0$
to $\pi_1^{\boldsymbol{e}_2} \cap \pi_3^{\boldsymbol{e}_2}: x = z \ge 0, (y - 1) = 0$;
the shifted ray of
$\pi_2 \cap \pi_3: y = z \ge 0, x = 0$ moves from
$\pi_2^{\boldsymbol{e}_1} \cap \pi_3^{\boldsymbol{e}_1}: y = z \ge 0, (x - 1) = 0$
to $\pi_2^{\boldsymbol{e}_2} \cap \pi_3^{\boldsymbol{e}_2}: (y - 1) = z \ge 0, x = 0$.
Notice that both stay on the plane $x + y - z = 1$ 
which is $\pi_3^{\boldsymbol{e}_1} = \pi_3^{\boldsymbol{e}_2}$.
\end{itemize}
It follows that the part of $U$ satisfying $x + y - z > 1$ 
is covered by the  Minkowski sum of $W$ and the line segment
on  $x + y = 1, z = 0$  from $\boldsymbol{e}_1$ to $\boldsymbol{e}_2$
(which is a side of the triangle $V$).

Symmetrically, after sliding $W$ with its origin
from $\boldsymbol{e}_2$  to $\boldsymbol{e}_3$
along the line $y + z = 1, x = 0$
we get the parallel tranport from $W^{\boldsymbol{e}_2}$ to 
$W^{\boldsymbol{e}_3}$.
Also  after sliding $W$ with its origin
 from $\boldsymbol{e}_3$ back to $\boldsymbol{e}_1$
along the line segment $x + z = 1, y = 0$
we get the parallel tranport from
$W^{\boldsymbol{e}_3}$ to $W^{\boldsymbol{e}_1}$.
After these, the only subset in $U$ that is left uncovered by shifting
copies of $W$ is
$U \cap \{(x, y, z) \ |
\left\{\begin{smallmatrix}
- x + y + z \le 1 \\
x - y + z \le 1 \\
x + y - z \le 1
\end{smallmatrix}\right.\} = \{(x, y, z)  \in \mathbb{R}^3_{>0} \ |
\left\{\begin{smallmatrix}
x + y + z \ge 1\\
- x + y + z \le 1 \\
x - y + z \le 1 \\
x + y - z \le 1
\end{smallmatrix}\right.\}$ --- a \emph{tetrahedron} (\figref{fig:tetrahedron}).
However this subset can be covered by the rays $\{ \boldsymbol{v} + 
\lambda (1, 1, 1) \ |\ \boldsymbol{v} \in V,  \lambda >0 \}$.
Note that $ \lambda (1,1,1) \in W$ for all $\lambda >0$.
This completes the proof.
\end{proof}

\begin{lemma}\label{lem:pm-easy_tight}
Suppose $f$ is the constraint function of a 4-ary matchgate with $M(f) = \left[\begin{smallmatrix} d_1 & 0 & 0 & a_1 \\ 0 & b_1 & c_1 & 0 \\ 0 & c_2 & b_2 & 0 \\ a_2 & 0 & 0 & d_2 \end{smallmatrix}\right]$.
Then $f \in \mathcal{S}^\textup{E}_{\le^2}$.
In particular, if $f$ satisfies arrow reversal symmetry, $f \in \textup{\texttt{SQ-SUM}}$.
\end{lemma}
\begin{remark}
The last part $d_1d_2 \ \le\ a_1a_2 + b_1b_2 + c_1c_2$ was proved in \cite[Lemma 56]{BULATOV201711}.
The proofs for other three parts are symmetric and similar to the proof for the last part.
For completeness, here we give the proof for the first part $a_1a_2 \ \le\ b_1b_2 + c_1c_2 + d_1d_2$.
\end{remark}
\begin{proof}
Consider a 4-ary matchgate $\Gamma$ with constraint function $f$.
Given that $M(f) = \left[\begin{smallmatrix} d_1 & 0 & 0 & a_1 \\ 0 & b_1 & c_1 & 0 \\ 0 & c_2 & b_2 & 0 \\ a_2 & 0 & 0 & d_2 \end{smallmatrix}\right]$,
$a_1a_2 \ \le\ b_1b_2 + c_1c_2 + d_1d_2$ is equivalent as
\begin{equation}\label{eqn:tight}
f(0011)f(1100) \le f(0110)f(1001) + f(0101)f(1010) + f(0000)f(1111).
\end{equation}
Let $I = \{i_1, i_2, i_3, i_4\}$ be the set of dangling edges of $\Gamma$.
For $X \subseteq I$, let $M_X$ denote the set of perfect matchings that include dangling edges in $X$ (by assigning them $1$) and exclude dangling edges in $I \setminus X$ (by assigning them $0$).
We exhibit an injective map
\[
\mu: M_{\{i_1, i_2\}} \times M_{\{i_3, i_4\}} \rightarrow 
[ M_{\{i_2, i_3\}} \times M_{\{i_1, i_4\}} ]
 \bigcup 
[ M_{\{i_2, i_4\}} \times M_{\{i_1, i_3\}} ]
 \bigcup 
[ M_\emptyset \times M_I ]
\]
which is weight-preserving in the sense that
for matchings $m_1, m_2, m_3, m_4$ with $\mu(m_1, m_2) = (m_3, m_4)$, we have $w(m_1) w(m_2) = w(m_3)w(m_4)$. The existence of $\mu$ implies (\ref{eqn:tight}).

Given $(m_1, m_2) \in M_{\{i_1, i_2\}} \times M_{\{i_3, i_4\}}$, consider $m_1 \oplus m_2$ and note that this is a collection of cycles together with two paths.
Let $\pi$ be the path connecting the dangling edge $i_1$ to some other dangling edge; let $\pi'$ be the path connecting the remaining two dangling edges.
Let $m_3 := m_1 \oplus \pi$ and $m_4 := m_2 \oplus \pi$.
Then we have the following
\begin{itemize}
\item
If $\pi$ connects $i_1$ to $i_2$, then $m_3 \in M_\emptyset$ and $m_4 \in M_I$;
\item
If $\pi$ connects $i_1$ to $i_3$, then $m_3 \in M_{\{i_2, i_3\}}$ and $m_4 \in M_{\{i_1, i_4\}}$;
\item
If $\pi$ connects $i_1$ to $i_4$, then $m_3 \in M_{\{i_2, i_4\}}$ and $m_4 \in M_{\{i_1, i_3\}}$.
\end{itemize}
The construction is invertible, since
if $(m_3, m_4)$ is in the image of the above mapping,
then $m_3 \oplus m_4 = m_1 \oplus m_2$. From $m_1 \oplus m_2$,
 we can recover $\pi$ 
(as the unique path that
connects $i_1$ to one of the other dangling edges in $\{i_2, i_3, i_4\}$).
Then we can recover
    $m_1$ and $m_2$
as $m_3 \oplus \pi$ and $m_4 \oplus \pi$ respectively. 
Therefore, $\mu: (m_1, m_2) \rightarrow (m_3, m_4)$ is an injection.

To see that $\mu$ is weight-preserving, observe that the each of the edges in $\pi$ appears in exactly one of $m_1$ and $m_2$ and in exactly one of $m_3$ and $m_4$ and that $m_i \setminus \pi = m_{i+2} \setminus \pi$ for $i \in \{1, 2\}$.
Hence,
\[
w(m_1)w(m_2) = \prod_{e \in m_1 \setminus \pi} w_e
\cdot \prod_{e \in m_2 \setminus \pi} w_e
\cdot \prod_{e \in \pi} w_e = \prod_{e \in m_3 \setminus \pi} w_e 
\cdot \prod_{e \in m_4 \setminus \pi}  w_e 
\cdot  \prod_{e \in \pi}  w_e  = w(m_3)w(m_4).
\]
\end{proof}

\bibliography{reference}{}
\bibliographystyle{alpha}

\clearpage
\appendix
\section*{Appendix}

\section{}\label{app:complexity_table}
\newcommand{\ra}[1]{\renewcommand{\arraystretch}{#1}}
\begin{table*}[h!]
\centering
\caption[Caption for LOF]{Approximation complexity of the eight-vertex model with $(a, b, c, d) \not\in \texttt{$d$-SUM}$\protect\footnotemark.}
\ra{1.7}
\begin{tabular}{@{}l@{\phantom{abc}}l@{\phantom{abc}}l@{}}
\toprule
 & $d = 0$ & $d > 0$ \\
\midrule
$a=b=c=0$ & P-time computable (trivial) & P-time computable (trivial) \\
$a=b=0, c>0$ & P-time computable (trivial) & \pbox{20cm}{$c=d$: P-time computable~\cite{DBLP:journals/corr/CaiF17} \\ $c\neq d$: NP-hard~\cite{DBLP:journals/corr/abs-1811-03126}} \\
$a=0, b, c >0$ & NP-hard~\cite{DBLP:journals/corr/abs-1811-03126} & \#\textsc{PerfectMatchings}-hard (in this paper)\\
$a, b, c > 0$ & NP-hard~\cite{doi:10.1137/1.9781611975482.136} & \#\textsc{PerfectMatchings}-hard (in this paper)\\
\bottomrule
\end{tabular}
\end{table*}

\footnotetext{There is a symmetry among $a, b, c$ for the eight-vertex model on general (not necessarily planar) 4-regular graphs, so for simplicity we assume $a \le b \le c$.}

\section{}\label{app:star_condition}

Given a constraint function $f$ with $M(f) = \left[\begin{smallmatrix} d & 0 & 0 & a \\ 0 & b & c & 0 \\ 0 & c & b & 0 \\ a & 0 & 0 & d \end{smallmatrix}\right]$ and $a, b, c \ge 0, d > 0$,
we show how to obtain a 4-ary construction with constraint function $\hat{f}$ such that $M(\hat{f}) = \left[\begin{smallmatrix} \hat{d} & 0 & 0 & \hat{a} \\ 0 & \hat{b} & \hat{c} & 0 \\ 0 & \hat{c} & \hat{b} & 0 \\ \hat{a} & 0 & 0 & \hat{d} \end{smallmatrix}\right]$ and $\hat{a}, \hat{b}, \hat{c}, \hat{d}$ satisfy the condition  $0 < \hat{d} \le \hat{a} \le \hat{b} \le \hat{c} \le \frac{3}{2}\hat{d}$, using copies of the constraint function $f$.
(When used in the context of
  \lemref{lem:pm-hard_general}, the number of copies of $f$ used is
a  constant.)
By a scaling we  get 
\[1 \le \hat{d} \le \hat{a} \le \hat{b} \le \hat{c} \le \frac{3}{2}\hat{d}.\]

We first deal with the case where one of $a, b, c$ is equal to 0.
By relabeling we may assume
$b=0$. Then we use the construction $G_1$ in the proof 
of \lemref{lem:pm-hard_general} to get a signature
with parameters $(a_1, b_1, c_1, d_1) = (ac, d^2, a^2, cd)$ with no 
zero entries.

In the following we may assume all entries are positive $a, b, c, d > 0$.

Our first step is to satisfy (\ref{eqn:condition}).
If all four entries are equal $a=b=c=d$ then we are done (the signature
is also in the affine family and thus the problem is also P-time tractable.)
So we may assume  at least one of $a, b, c \neq d$. By relabeling we 
may assume $a \neq d$.
We may also assume at least one of the following three displayed inequalities
does not hold, and by relabeling we may assume $a +d < b + c$. Indeed, if all
three inequalities
\[ a + d \ge b + c, ~~~~ b + d \ge a + c, ~~~~ c + d \ge a + b\]
hold,
then by connecting 
two copies of $f$ using the 4-ary construction in \figref{fig:2-chain} and 
relabeling edges, we get $(a', b', c', d') =
(2bc, a^2 + d^2, b^2 + c^2, 2ad)$.
This signature has $a' + d' < b' + c'$.
%

Under a holographic transformation (in the proof of
 \lemref{lem:pm-hard_holant}),
in the edge-2-coloring view
 $\left[\begin{smallmatrix} d & 0 & 0 & a \\ 0 & b & c & 0 \\ 0 & c & b & 0 \\ a & 0 & 0 & d \end{smallmatrix}\right]$ becomes $\frac{1}{2}
\left[\begin{smallmatrix} a + b + c + d & 0 & 0 & - a + b + c - d \\ 0 & a - b + c - d & a +b - c - d & 0 \\ 0 & a +b - c - d & a - b + c - d & 0 \\ - a + b + c - d & 0 & 0 & a + b + c + d \end{smallmatrix}\right] = 
P \cdot \left[\begin{smallmatrix} a+d & 0 & 0 & 0 \\ 0 & c-b & 0 & 0 \\ 0 & 0 & a-d & 0 \\ 0 & 0 & 0 & c+b \end{smallmatrix}\right] \cdot P^{-1}$, where $P = \left[\begin{smallmatrix} 1 & 0 & 0 & 1 \\ 0 & 1 & 1 & 0 \\ 0 & -1 & 1 & 0 \\ -1 & 0 & 0 & 1 \end{smallmatrix}\right]$.
Connect $k$ copies of $f$ using the 4-ary construction in \figref{fig:2-chain}, we get the constraint function $f'$ with $M(f') = \left[\begin{smallmatrix} d' & 0 & 0 & a' \\ 0 & b' & c' & 0 \\ 0 & c' & b' & 0 \\ a' & 0 & 0 & d' \end{smallmatrix}\right]$ which in the edge-2-coloring view
becomes $P \left[\begin{smallmatrix} a'+d' & 0 & 0 & 0 \\ 0 & c'-b' & 0 & 0 \\ 0 & 0 & a'-d' & 0 \\ 0 & 0 & 0 & c'+b' \end{smallmatrix}\right] P^{-1}$. %
By construction, it is also equal to
\[P\left[\begin{smallmatrix} a+d & 0 & 0 & 0 \\ 0 & c-b & 0 & 0 \\ 0 & 0 & a-d & 0 \\ 0 & 0 & 0 & c+b \end{smallmatrix}\right]^{k}P^{-1} = 
P\left[\begin{smallmatrix} (a+d)^k & 0 & 0 & 0 \\ 0 & (c-b)^k & 0 & 0 \\ 0 & 0 & (a-d)^k & 0 \\ 0 & 0 & 0 & (c+b)^k \end{smallmatrix}\right]P^{-1}.
\]
It follows that
\[a' =\frac{(a+d)^k+ (a-d)^k}{2},~~
  b' =\frac{(b+c)^k - (c-b)^k}{2},~~
  c' =\frac{(b+c)^k + (c-b)^k}{2},~~
  d' = \frac{(a+d)^k- (a-d)^k}{2}.\] 
Since $b+ c$ is strictly larger than $a+d$,  and also
$b+c = |b+c| > |c-b|$ and $|a-d|$,
 one can see that for a large $k$, we can 
get $f'$ such that $b' \approx c' \gg a', d'$, and all entries are positive.
If furthermore $d' \le a'$, then we have arrived at a signature
$f'$ satisfying (\ref{eqn:condition}) with a possible relabeling.
If $d' > a'$,
using the construction in \figref{fig:2-chain} with two copies of $f'$ (and two $(\neq_2)$ in the middle as we are in the orientation view),
we can get $f''$ with
$M(f'') = M(f') \cdot N \cdot M(f') = \left[\begin{smallmatrix} 2a'd' & 0 & 0 & {a'}^2+{d'}^2 \\ 0 & 2b'c' & {b'}^2+{c'}^2 & 0 \\ 0 & {b'}^2+{c'}^2 & 2b'c' & 0 \\ {a'}^2+{d'}^2 & 0 & 0 & 2a'd' \end{smallmatrix}\right]$.
The four parameters for $f''$ is $(a'', b'', c'', d'') = ({a'}^2 + {d'}^2, 2b'c', {b'}^2 + {c'}^2, 2a'd')$ which satisfies (\ref{eqn:condition})
(with a possible relabeling) when $b', c' \gg a', d'$.

Now that we have a constraint function $f'$ (or $f''$ if the last
step is needed)
 satisfying (\ref{eqn:condition}), we obtain $\hat{f}$ using 4-ary constructions with $f'$ (or $f''$ respectively).
This is done by repeatedly applying $G_2$ (defined in \lemref{lem:pm-hard_general}) in which each round starts with the
function constructed from the previous round, and possibily relabeling the 
input values $(a, b, c, d)$,
so that the construction $G_2$ 
starts with some $(a, b, c, d)$  satisfying
 (\ref{eqn:condition}).
Note that for any $(a, b, c, d)$ satisfying (\ref{eqn:condition}) and $c > \frac{3}{2}d$, $G_2$ shrinks the distance between $c$ and $d$ by 
at least a constant (after normalization).
This can be checked by $\frac{c}{d} - \frac{c_2}{d_2} = \frac{\left(\frac{c}{d} + 1\right)\left(\frac{c}{d} - 1\right)}{2 \cdot \frac{c}{d}} > \frac{5}{12}$ when $c > \frac{3}{2}d$.
Therefore, starting with $f'$ and iterative applying $G_2$ 
(constantly many times), we can obtain $\hat{f}$ so that $1 \le \hat{d} \le \hat{a} \le \hat{b} \le \hat{c} \le \frac{3}{2}\hat{d}$.

\end{document}